\documentclass[%
 aip,
 amsmath,amssymb,
 reprint,%
]{revtex4-1}

\usepackage{graphicx}
\usepackage{dcolumn}

\usepackage[english]{babel}
\usepackage[utf8x]{inputenc}
\usepackage{amsmath}
\usepackage{graphicx}
\usepackage[colorinlistoftodos]{todonotes}
\def \bF {\pmb{F}}
\def \bH {\pmb{H}}
\def \bG {\pmb{G}}
\def \bx {\pmb{x}}
\def \bX {\pmb{X}}

\usepackage{graphicx}
\usepackage{dcolumn}
\usepackage{bm}
\usepackage{xcolor}
\usepackage[normalem]{ulem}
\usepackage[title]{appendix}
\usepackage{comment}
\usepackage{float}
\usepackage{mathtools}

\usepackage{hyperref}

 \usepackage{amssymb}
\usepackage{amsthm}
\usepackage{stmaryrd}
\usepackage{amsfonts}
\usepackage{amstext}
\newtheorem{theorem}{Theorem}
\newtheorem{lemma}{Lemma}
\newtheorem{definition}{Definition}
\newtheorem{remark}{Remark}
\newtheorem{corollary}{Corollary}
\usepackage{subfigure}
\newcommand{\bb}[1]{%
	\textbf{\textcolor{blue}{#1}}%
}

\usepackage{epstopdf}
\newcolumntype{C}[1]{>{\centering\arraybackslash$}m{#1}<{$}}
\newlength{\mycolwd}                                
\begin{document}

\title{Pinning control of networks:  dimensionality reduction through simultaneous block-diagonalization of matrices}

\author{Shirin Panahi}
\affiliation{University of New Mexico, Albuquerque, NM, US 80131}

\author{Matteo Lodi}

\affiliation{University of Genoa, Via Opera Pia 11A, 16154, Genova, Italy}

\author{Marco Storace}
\affiliation{University of Genoa, Via Opera Pia 11A, 16154, Genova, Italy}



\author{Francesco Sorrentino}
\email{fsorrent@unm.edu}
\affiliation{University of New Mexico, Albuquerque, NM, US 80131}


\begin{abstract}
 In this paper, we study the network pinning control problem in the presence of two different types of coupling: (i) node-to-node coupling among the network nodes and (ii) input-to-node coupling from the source node to the `pinned nodes'. Previous work has {mainly} focused on the case that (i) and (ii) are of the same type. We decouple the stability analysis of the target synchronous solution into subproblems of the lowest dimension by using the techniques of simultaneous block diagonalization (SBD) of matrices. Interestingly, 
we obtain two different types of blocks, driven and undriven. The overall dimension of the driven blocks is equal to the dimension of an appropriately defined controllable subspace, while all the remaining undriven blocks are scalar. Our main result is a decomposition of the stability problem into four independent sets of equations, which we call quotient controllable, quotient uncontrollable, redundant controllable, and redundant uncontrollable. {Our analysis shows that} the number and location of the pinned nodes affect the number and the dimension of each set of equations. {We also observe that in a large variety of 
complex networks}, stability of the target synchronous solution is \emph{de facto} only determined by a single quotient controllable block.
\end{abstract}

\maketitle

\begin{quotation}
{In this paper we consider dynamical networks
formed of coupled identical oscillators and use
pinning control to synchronize all of the network
{nodes} on a given target time evolution.}
The
problem of stability of the entire network about
the target solution is then studied by linearizing the network dynamics about the target solution. Different from previous work, we focus on
the general case that the node-to-node coupling
among the network nodes is different from the
input-to-node coupling from the source node to
the ‘pinned nodes’. The stability problem is then
decoupled into the stability of independent subsystems of the lowest dimension.
Our analysis is relevant to the analysis of several {complex networks} for which we see
that stability of the target synchronous solution is often
determined by the maximum Lyapunov exponent associated with only
one subsystem that we call `quotient controllable'.
\end{quotation}

\section{Introduction}
\label{s:intro}
Gaining control over collective network dynamics is a hot topic in both graph and control theory. In particular, pinning control is a feedback control strategy largely used for imposing synchronization or consensus in complex dynamical networks \cite{wu2008cluster,liu2011cluster,su2012decentralized,yu2013synchronization,liu2014synchronization,wang2015pinning,liu2015synchronization,delellis2018partial, belykh2008cluster, belykh2000hierarchy,chen2022pinning,tang2013,chen2014pinning,liu2015pinning,orouskhani2016optimizing,chen2018pinning}. Specifically, one or more virtual leaders (the so-called \textit{sources}) are added to the network and define its desired trajectory. Each source directly controls only a small fraction of the network nodes (the pinned nodes), by exerting a control action that is a function of the pinning error vector, whose $i$-th component is given by the difference between the output of the considered source and the output of the $i$-th node. {Pinning control of time varying networks has been studied in \cite{dariani2011effect,feng2016cluster,lin2021pinning,shi2021synchronization}.} Our ability to impose a desired synchronous solution to a dynamical network can have important applications in different fields of science, such as in physical \cite{an:2011}, social \cite{kan:2015}, multi-agent \cite{trentelman:2013}, and biological networks  \cite{russo:2009} and pinning control is a widely adopted solution \cite{su:2013}. 

Here we consider the problem of pinning control of undirected 
networks of dynamical systems, for the case in which the node-to-node connectivity is different from the connectivity exerted on the pinned nodes, which is relevant to a variety of realistic scenarios and engineering applications. As an example, imagine a biological network that one wants to synchronize on a specific time evolution, by pinning some of the network nodes. However, the type of forcing one may be able to exert on the network is likely going to be different from the biological node-to-node interactions between the network nodes. This motivates the study of a problem in which the node-to-node coupling among the network nodes is different from the coupling exerted on the pinned nodes. Similar versions of this problem have been previously investigated in \cite{wu2009,barajas2018,vega2018,alanis2021} using a Lyapunov function (V-stability) which provides a sufficient stability condition. In this paper, we investigate stability using linearization, which provides both necessary and sufficient conditions.

We study the stability of the target synchronous solution in an undirected 
network of coupled dynamical systems, subject to the control action of one source node. Our goal is to reduce (through proper transformations) the stability analysis into simpler problems, which can be analyzed independently of one another. 
This problem presents mathematical challenges that require the introduction of a specific formalism; in particular, we use the techniques for simultaneous block diagonalization of matrices (SBD) \cite{uhlig:1973,maehara2010numerical, murota2010numerical} which allow the reduction of the stability problem in problems of lower dimension. Previous work has successfully applied this approach to both the stability of complete synchronization \cite{Ir:So} and cluster synchronization \cite{zhang2020}. However, no characterization was provided about the dimensions of the independent sets of equations in which the stability problem is reduced.
The dimension of these sets of equations (which are related to the diagonal blocks of some matrices) is the lowest and we show that it may vary based on the structure of the network and the choice of the pinned nodes. Once these sets/blocks are found, a maximum Lyapunov exponent (MLE) can be associated to each of them in terms of a Master Stability Function. The target synchronous solution is stable only when the MLEs corresponding to all the sets/blocks are negative, which provides a simple criterion for assessing stability.

In summary, by combining MSF and SBD we study the stability of the target solution in the considered networks subject to pinning control; this leads to a decomposition of the stability problem {into 
four} different types of independent equations, which we call quotient controllable, quotient uncontrollable, redundant controllable,  and  redundant  uncontrollable.
{The insight we provide on both the sizes and `roles' of the blocks is the main contribution of this paper. As stated before, though there is previous work on the SBD reduction, to the best of our knowledge, no paper has explained the `reason' for the blocks resulting from application of the SBD decomposition.}

To carry out the stability analysis, we resort to two alternative transformations: one ($T$) based on the SBD decomposition and one ($\hat{T}$) based on the concepts of controllable subspace and equitable clusters. 
We are the first ones to apply the SBD approach to the pinning control problem and in so doing we establish a connection between the blocks in which the stability problem is reduced and the particular choice of the network connectivity and of the pinned nodes. 
In particular, we show that the sizes of the blocks provided by a finest SBD are the same as for the blocks generated by the transformation $\hat{T}$, which has a clear interpretation in terms of controllability and quotient graphs. {We prove that the transformation $\hat{T}$ also provides a finest SBD.} 
This is another important contribution of this paper.


The rest of the paper is organized as follows. In Sec.\ \ref{s:Pin} we  introduce the problem of pinning control of networks. In Sec.\ \ref{stability} a stability analysis is derived for a network with pinned nodes. The solution to the stability problem consists of two steps presented in Secs. \ref{sec:SBD} and \ref{results}. 
First, we use the SBD method to reduce the dimension of the stability problem. Then, by using two appropriately defined transformation matrices, we obtain the transformation $\hat{T}$ which decouples the stability problem into four independent blocks. We then compare the application of the SBD transformation ($T$) with the transformation $\hat{T}$. The effects of the selection of different pinned nodes on multiple driven blocks is studied in Sec.\ VI. Application of the theory to larger complex networks is considered in Sec.\ VII. Finally, the conclusions are drawn in Sec. \ref{conclusion}.

\section{Problem Definition: Pinning control of networks}\label{s:Pin}
We consider an undirected network of $N$ coupled dynamical systems, {which evolve in time according to the following equation:}
\begin{equation}
\dot{\bx}_{i}(t) = \bF(\bx_{i}(t))+\sum_{j=1}^{N}A_{ij}[\bG(\bx_{j}(t))-\bG(\bx_{i}(t))] \quad  i=1, \cdots, N
\label{Eq1}
\end{equation}
where $\bx_i$ is the $m$-dimensional state of node $i$, $\bF : R^{m} \rightarrow R^{m}$ is the function that governs the dynamics of each node {when isolated}, and $\bG:R^{m} \rightarrow R^{m}$ is the node-to-node coupling function.  The network topology is described by the adjacency matrix $A$. If there is a connection between the nodes $i$ and $j$ then $A_{ij}=A_{ji}=1$ otherwise $A_{ij}=A_{ji}=0$.  Eq.\ \eqref{Eq1} can be rewritten as:
\begin{equation}
\dot{\bx}_{i}(t) = \bF(\bx_{i}(t))+\sum_{j=1}^{N}L_{ij}\bG(\bx_{j}(t)) \quad  i=1, \cdots, N,
\label{Eq2}
\end{equation}
where the Laplacian matrix $L=A-D$ and $D$ is a diagonal matrix, such that $D_{ii}=\sum_j A_{ij}$.
This network allows a completely synchronized  solution $\bx_1(t)=\bx_2(t)= \cdots =\bx_n(t)=\bx_s(t)$, which obeys,
\begin{equation}
\dot{\bx}_{s}(t)=\bF(\bx_{s}(t)).
\label{Eq3}
\end{equation}
A relevant question that links control theory to graph theory is: how can control inputs be introduced to force the network state to the target synchronous state? 

{Equation (3) allows for a} number of different {synchronous} solutions {determined by its initial condition.}
The problem studied in pinning control is how control inputs generated by source nodes can be designed to enforce stability of a particular `target' synchronous solution, $\textbf{x}_t(t)$ produced by the initial condition $\bx_t^0$,
\begin{equation}
\dot{\bx}_{t}(t)= \bF(\bx_{t}(t)), \quad \quad
\bx_t(0)= \bx_t^0.
\label{target}
\end{equation}

In order to address this problem, we introduce control inputs $\textbf{u}_i(t)$ as pinning control signals,
\begin{subequations}\label{Eq4}
\begin{equation}\label{Eq4a}
    \dot{\bx}_{i}(t) = \bF(\bx_{i}(t))+\sum_{j=1}^{N}L_{ij}\bG(\bx_{j}(t))+\textbf{u}_i(t),
\end{equation}
\begin{equation}\label{Eq4b}
    \textbf{u}_i(t)=  \gamma r_i[\bH(\bx_{t}(t))-\bH(\bx_{i}(t))],
\end{equation}
\begin{equation*}
    \quad  i=1, \cdots, N
\end{equation*}
\end{subequations}
where the binary scalar $r_i=1$ ($r_i=0$) if node $i$ is pinned ({not pinned}), and the scalar $\gamma>0$ measures the strength of the {control} coupling. Here, $\bH:R^{m} \rightarrow R^{m}$ is the source-to-pinned-node coupling function.
Note that the control action is only directly active on the pinned nodes. For instance, consider the network shown in Fig. \ref{Network with Pinned Nodes}. The network consists of $N=11$ nodes coupled via blue edges. A source node that provides the  target trajectory is shown in red. The control input  is directly applied (red directed edges) from the source node to the pinned nodes  $1$, $2$, and $3$, which are shown in black.
\begin{definition}\textbf{Network with inputs.}
A network with inputs is represented by a pair of $N\times N$ matrices $L$ and $R$, where the Laplacian matrix $L$ describes the network connectivity and the diagonal matrix $R$ is such that $R_{ii}=r_i, i=1,..,N$. We call $\mathcal{V}$ the set of the $N$ network nodes, $\mathcal{V}_P$ the set of the $s$ pinned nodes, and $\mathcal{V}_{NP}$ the set of $\tau=N-s$  non pinned nodes, $\mathcal{V}_P \cup \mathcal{V}_{NP}=\mathcal{V}$ and $\mathcal{V}_P \cap \mathcal{V}_{NP}= \emptyset$.
\end{definition}

\begin{definition}\textbf{Extended network.} \label{def extend network}
Given a network with inputs, its extended network is represented by the $(N+1) \times (N+1)$ directed Laplacian matrix $\tilde{L}$, \cite{sorrentino2007} defined below, 
\begin{equation}
    \tilde{L}=\begin{pmatrix}
  (L-R)_{N\times N} & \pmb{r}
  \\ 
  \mathbf{0}_{1\times N} & 0
  \end{pmatrix}
  \label{Lr}
\end{equation}
where the {column} vector $\pmb{r}=[r_1,r_2,...,r_N]^T$ and  $\mathbf{0}_{1\times N}$ is the zero row-vector of dimension $N$. 
\end{definition}

\begin{figure}[ht]
\centering
\includegraphics[scale=.4]{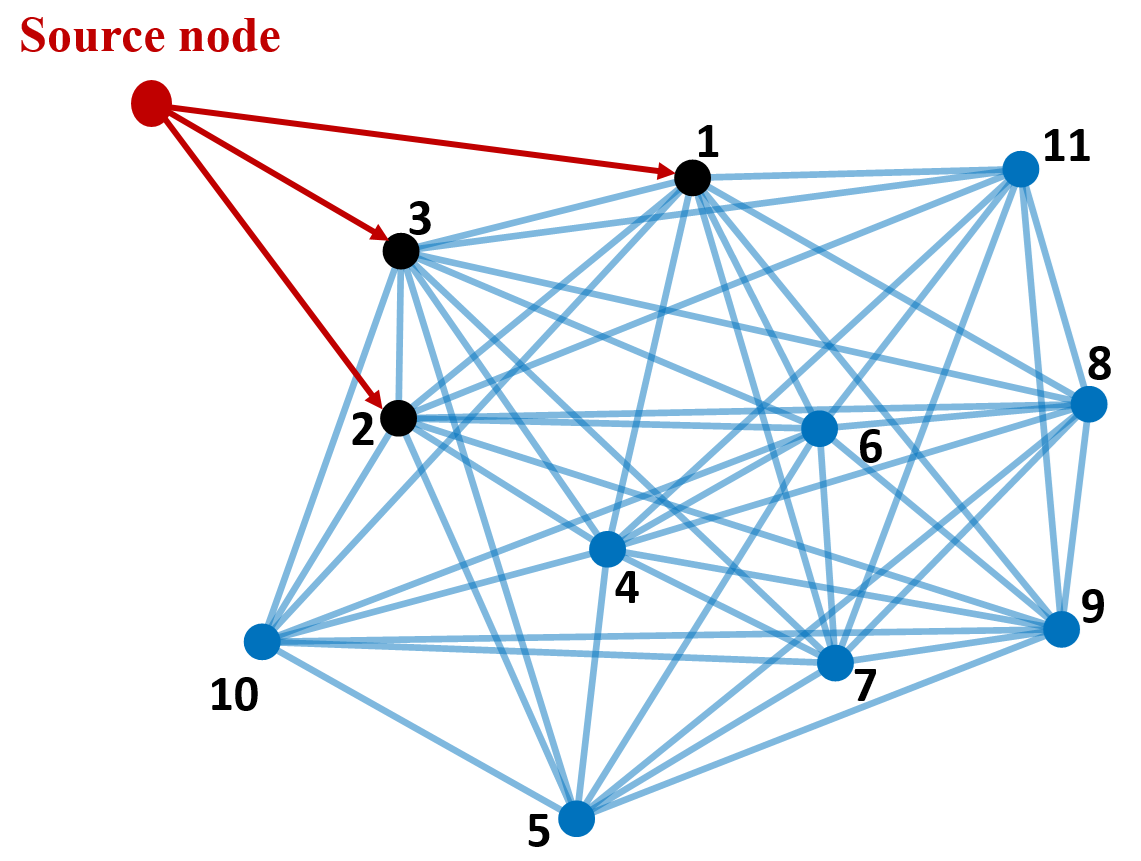}
\caption{A network of {$N=11$} nodes subject to pinning control. The source node is shown in red. The pinned nodes (i.e., the nodes that receive the source signal) are shown in black. The remaining network nodes are in blue. The node-to-node network connections are in blue. Red connections carry the target solution $\bx_{t}(t)$ (Eq. \eqref{target}) from the source node to the pinned nodes.}
\label{Network with Pinned Nodes}
\end{figure}

Previous work (see e.g., {\cite{sorrentino2007,sorrentino2007effects}}) has 
{analyzed stability of the target synchronous  solution} by considering pinned nodes with the same coupling function as the node-to-node coupling function, i.e., with $\bH=\bG$. Here instead we consider a generalization where the effect of the source node's time evolution on the pinned nodes is given by a different coupling function.

\section{Stability Analysis}\label{stability}
When all the $\pmb{x}_i(t)$ converge on the target solution $\pmb{x}_t(t)$, the control inputs $\textbf{u}_i(t)$ converge to zero.
To study the stability of the synchronous solution, a small perturbation $\delta \pmb{x}_i=(\pmb{x}_i-\pmb{x}_t)$ is considered \cite{Pe:Ca}. Linearization of Eq.\ \eqref{Eq4} about the target synchronous solution $\bx_1(t)=\bx_2(t)= \cdots =\bx_n(t)=\bx_t(t)$  yields, {see also \cite{sorrentino2007}}, 
\begin{equation}
\begin{array}{clc}
    \delta\dot{\bx}_{i}(t) = D\bF(\bx_{t}(t))\delta \bx_{i}(t)+\sum_{j=1}^{N}L_{ij}D\bG(\bx_{t}(t))
    \delta \bx_{j}(t)\\
    -\gamma r_i D\bH(\bx_{t}(t))\delta \bx_{i}(t),
    \end{array}
\label{Eq5}
\end{equation}
$i=1, \cdots, N$, where $D$ is the Jacobian operator.  

By stacking all perturbation vectors together in one vector $\delta \pmb{X}=[\delta \pmb{x}_1^T,\delta \pmb{x}_2^T,…,\delta \pmb{x}_N^T ]^T$, Eq.\ \eqref{Eq5} can be rewritten as
\begin{equation}
\begin{array}{clc}
\delta \dot{\bX}(t) = [ I_{N} \otimes D\bF(\bx_{t}(t)) + L\otimes D\bG(\bx_{t}(t))\\
-\gamma R \otimes D\bH(\bx_{t}(t)) ] \delta\bX(t),
\end{array}
\label{Eq11}
\end{equation}
where $\otimes$ is the Kronecker product or direct product. Our goal is to break the $(N \times m)$-dimensional Eq.\ \eqref{Eq11} into a set of independent lower-dimensional equations, thus performing a \textit{dimensional reduction}. An inherent difficulty is due to the fact that, with the exception of a few specific cases \cite{HYP}, it is not possible to simultaneously diagonalize both matrices $R$ and $L$. Instead, in what follows we seek a transformation that decouples the set of Eqs.\ \eqref{Eq11}, by simultaneously block diagonalizing $R$ and $L$ \cite{Ir:So}.

\section{Dimensionality reduction of the stability problem through SBD}\label{sec:SBD}

\begin{definition}{\textbf{Simultaneous Block Diagonalization of Matrices (SBD) {\cite{uhlig:1973,maehara2010numerical, murota2010numerical}}.}}
Given a set of {$q$ square $N$-dimensional} matrices $M_1, M_2, \cdots, M_q$, an SBD transformation is an orthogonal square matrix $T$ (transformation matrix) with dimension $N$ such that
\begin{equation}
T^{-1} M_k T =\oplus_{j=1}^{l}B^{k}_{j}, \quad k=1,2, \cdots, q,
\end{equation}
where the symbol $\oplus$ is the direct sum of matrices, $l$ denotes the number of blocks,  each block $B^{k}_{j}$ is a square matrix with dimension $b_j$ and $\sum_{j=1}^{l}b_j=N$.
\end{definition}
\begin{definition}{\textbf{Finest SBD.}}
A finest SBD is an SBD for which the resulting blocks can not be further refined \cite{uhlig:1973}. In particular this means that the size of the blocks can not be made smaller and that the number of blocks, $l$, is largest among all SBDs.
\end{definition}

\begin{remark}
    A finest SBD is unique only in the number and sizes of the irreducible blocks \cite{uhlig:1973}. 
\end{remark}

\begin{remark}
 The blocks resulting from a finest SBD are also matrices that belong to irreducible matrix $*$-algebras according to the structure theorem for matrix $*$-algebras.\cite{maehara2010numerical, murota2010numerical}
\end{remark}

Application of the SBD technique to complete synchronization of networks was first proposed in Ref.\ \cite{Ir:So} and only recently applied to cluster synchronization in Ref.\ \cite{zhang2020}.

Once we obtain $T=SBD(R,L)$, the matrices $R$ and $L$ are transformed as follows,
\begin{subequations}\label{eq:directsumLR}
\begin{equation}
    T^{-1} L T=  L_T=\oplus_{j=1}^{l} \hat{L}_j, \label{eq:directsumL}
\end{equation}
\begin{equation}
    T^{-1} R T=  R_T=\oplus_{j=1}^{l} \hat{R}_j, \label{eq:directsumR}
\end{equation}
\end{subequations}
where $L_T$ is the transformed matrix $L$, $R_T$ is the transformed matrix $R$, and the irreducible block pairs $(\hat{L}_j,\hat{R}_j)$ have the same dimensions $b_j$, $j=1,...,l$. One of the main contributions of this paper is to investigate the dimensions of the block pairs $(\hat{L}_j,\hat{R}_j)$, as a result of the original network topology (through $L$) and the choice of the pinned nodes (through $R$).
\subsection{Finding the matrix P} \label{Sec:P}

{The procedure described in Ref. \cite{maehara2011algorithm} to find the finest SBD transformation, $T$, requires two steps. First, finding a matrix $P$ which commutes with each member of the set of matrices (here, $R$ and $L$), that is $PR=RP$ and $PL=LP$, is selected. Second, the transformation matrix $T$ is constructed to have columns corresponding to the eigenvectors of the commuting matrix $P$.}

Here, we describe the approach to compute the matrix $P$ that commutes with both $R$ and $L$.
Without loss of generality, we apply {the same} permutation to {both} matrices $R$ and $L$ such that the matrix $R$ can be rewritten,
\begin{equation}
{R}=
\begin{pmatrix} \label{Btilde}
  I_{s} & 0\\ 
  0 & \mathbf{0}_{N-s}
  \end{pmatrix},
\end{equation}
where $I_s$ is the identity matrix of {size} $s$, $\mathbf0_{N-s}$ is the {square} zero matrix of dimension $N-s$, and $s$ is the number of pinned nodes. 
\begin{lemma}\label{lemma P}
Any matrix $P$ that commutes with the matrix $R$ in \eqref{Btilde} has 
the following block-diagonal structure,
\begin{equation}
P=\begin{pmatrix} \label{Pblock}
  P_{1} & 0\\
  0 & P_{2}
\end{pmatrix}
\end{equation}
where the block $P_{1}$ has dimension $s$  and the block $P_2$ has dimension $N-s$.
\end{lemma}

\begin{proof}
First, we consider the matrix $R$ of Eq.\ \eqref{Btilde} and a generic matrix ${P}=\big(\begin{smallmatrix}
  P_{1} & P_{12}\\ 
  P_{21} & P_{2}
\end{smallmatrix}\big)$ with sub-blocks having the same dimensions as those of the matrix ${R}$. Then,
{from} the commutation equation $P{R}={R}P$ we {obtain}
\begin{equation}
\begin{array}{lcl}
\begin{pmatrix}
  P_{1} & P_{12}\\
  P_{21} & P_{2}
\end{pmatrix} \begin{pmatrix}
  I_s & 0\\ 
  0 & 0_{(N-s)}
\end{pmatrix}=\begin{pmatrix}
  I_s & 0\\ 
  0 & \mathbf{0}_{(N-s)}
\end{pmatrix} \begin{pmatrix}
  P_{1} & P_{12}\\
  P_{21} & P_{2}
\end{pmatrix}\\\\
\begin{pmatrix}
  P_{1} & 0\\
  P_{21} & \mathbf{0}_{(N-s)}
\end{pmatrix}=\begin{pmatrix}
  P_{1} & P_{12}\\
  0 & \mathbf{0}_{(N-s)}
\end{pmatrix}.
\end{array}
\label{Eq12}
\end{equation}
{Therefore,} to fulfill the commutation equation $PR=RP$, the sub-blocks $P_{12}$ and $P_{21}$ must be zeros and the matrix $P$ is block diagonal with diagonal-blocks $P_{1}$ and $P_{2}$ in which each block has the same dimension as the diagonal-blocks of matrix $R$. 
\end{proof}

{
\begin{corollary}
The transformation matrix $T$ also has the same block-diagonal structure,
\begin{equation}
T=\begin{pmatrix} \label{Tblock}
  T_{1} & 0\\
  0 & T_{2}
\end{pmatrix},
\end{equation}
where $T_1$ is the matrix of eigenvectors of $P_1$ and $T_2$ is the matrix of eigenvectors of $P_2$.
\end{corollary}}

\begin{lemma} \label{l2}
Application of the block diagonal transformation $T$ to the matrix ${R}$,  transforms ${R}$ back to itself, i.e., $T^{-1}{R} T=R_T={R}$.
\end{lemma}

\begin{proof}
By considering the matrix $R$ of Eq. \eqref{Btilde} and the matrix $T$ of Eq. \eqref{Tblock} and by using the fact that a block-diagonal matrix can be inverted block by block, we have
\begin{equation}
\begin{array}{lcl}
T^{-1}RT=
\begin{pmatrix} 
  T_{1}^{-1} & 0\\
  0 & T_{2}^{-1}
\end{pmatrix}
\begin{pmatrix} 
  I_s & 0\\
  0 & \mathbf{0}_{(N-s)}
\end{pmatrix}
\begin{pmatrix}
  T_{1} & 0\\
  0 & T_{2}
\end{pmatrix}=\\
\begin{pmatrix} 
  I_s & 0\\
  0 & \mathbf{0}_{(N-s)}
\end{pmatrix}.
\end{array}
\label{proof BT}
\end{equation}
\end{proof}

 The above lemma has important consequences. In fact, as we will see in Sec.\ \ref{ss:d&ud}, {this will lead to a decomposition of the stability problem into equations that can be of either one of two types: driven or undriven.} 
\begin{remark}
 Lemma \eqref{l2} shows that $R_T=R$. This is true independent of any permutation of the rows and columns of the matrix $R$. In the examples that follow  we will sometimes show the matrix $R_T$ in a form that corresponds to permutations of rows and columns of the matrix $R$ in Eq.\ \eqref{Btilde}.
\end{remark}

\subsection{Driven and Undriven blocks} \label{ss:d&ud}
As stated before, the {purpose of }using the SBD transformation is to break the stability problem of Eq.\ \eqref{Eq11} into a set of independent equations of lowest dimension. 

Due to the {block-diagonal} structure of both matrices $L_{T}$ and $R_{T}$ {and to Eq. \eqref{Btilde}}, they can be decoupled into the pairs $(L_{d},R_{d})$ and $(L_{ud},0_{N-c})$, 
\begin{equation}
L_{T}=\begin{pmatrix}
  L_{d} & 0\\
  0 & L_{ud}
\end{pmatrix}
\label{Eq18}
\quad \quad
R_{T}=\begin{pmatrix}
  R_{d} & 0\\
  0 & 0_{(N-c)}
\end{pmatrix},
\end{equation}
where the square matrices $L_d,R_d$ have size $c \geq s$ and the square matrix $L_{ud}$ has dimension $N-c$. {We remark that in general the blocks $L_d$ and $R_d$ are composed of smaller diagonal blocks. In particular,} if $c=s$, then $R_d=I_s$, otherwise $R_d$ is equal to a diagonal matrix with $s$ entries on the main diagonal equal to $1$ and $c-s$ entries on the diagonal equal to $0$. 
In what follows we refer to $(L_d,R_d)$ as the driven pair and $(L_{ud},0_{N-c})$ as the undriven pair. With an abuse of language, we also refer to $c$ as the dimension of the driven pair.

{Equation \eqref{Eq11}} can be {split into the following sets of equations, one driven and one undriven,}

\begin{subequations}\label{16}
\begin{equation}\label{driv}
\begin{split}
\dot{\pmb{\eta}}_{d}(t) = [ I_{c} \otimes D\bF(\bx_{s}(t)) + L_{d} \otimes D\bG(\bx_{s}(t)) \\
- \gamma R_{d} \otimes D\bH(\bx_{s}(t)) ]\pmb{\eta}_{d}(t),
\end{split}
\end{equation}
\begin{align}
\label{undriv}
\dot{\pmb{\eta}}_{ud}(t) &= \left[ I_{N-c} \otimes D\bF(\bx_{s}(t)) + L_{ud} \otimes D\bG(\bx_{s}(t)) \right]\pmb{\eta}_{ud}(t).
\end{align}
\end{subequations}

We then see that the original $mN$-dimensional problem of Eq. \eqref{Eq11} has been reduced {to }two independent lower-dimensional equations: the driven Eq.\ \eqref{driv} with dimension $mc$ and the undriven Eq.\ \eqref{undriv} with dimension $m(N-c)$.

{Next, we attempt to gain further insights into the physical meaning of the quantity $c$.}

\begin{definition}{\textbf{Dimension of the controllable subspace.}} Given {an} $N$-dimensional pair $(A,B)$, the Kalman controllability matrix {is} $K=[B,AB,A^2B,...,A^{N-1}B]$\cite{ogata2010modern}. The dimension of the controllable subspace is equal to the rank of the matrix $K$. The pair $(A,B)$ is completely controllable if 
{its} rank is equal to $N$.
\end{definition}

\begin{definition}{\textbf{
Transformation $\pmb{T_{c}}$.} \label{standard}} Given a $N$-dimensional canonical linear, time-invariant system given by the pair $(A,B)$, $A=A^T, B=B^T$, with dimension of the controllable subspace equal to $h$, the 
transformation $T_c(A,B)$ {splits} the system into a controllable subsystem and an uncontrollable subsystem \cite{kailath1980linear},
\begin{equation}
T_c^{-1} A T_c=\begin{pmatrix}
  A_C & 0\\
  0 & A_{U}
\end{pmatrix}
\quad \quad
T_c^{-1} B =\begin{pmatrix}
  B_C & 0 \\
  0  & 0
\end{pmatrix},
\end{equation}
where the controllable subsystem is given by the $h$-dimensional pair $(A_C,B_C)$ and the uncontrollable subsystem is given by the $(N-h)$-dimensional pair $(A_{U},0)$.
\end{definition}
{Note that the transformation applied to the matrix $A$ is a similarity transformation while the transformation applied to the matrix $B$ is not.}
\begin{theorem} The $c$ dimension of the driven pair $(L_d,R_d)$ is equal to $h$ the dimension of the controllable subspace of the pair $(L,R)$.
\end{theorem}
\begin{proof}
{We prove this by contradiction. First we assume $c < h$, i.e., the dimension $c$ of the driven pair $(L_d,R_d)$ is less than the dimension $h$ of the controllable subspace of the pair $(L,R)$. Since the dimension of the controllable subspace is invariant under any similarity transformation, then the controllable subspace of the pair $(L_T,R_T)$ cannot be less than the dimension of the controllable subspace of the pair $(L,R)$, which contradicts the assumption.
Then we assume $c > h$, which indicates that the pair $(L_d, R_d)$ is not controllable, since the dimension of the controllable subspace of the pair $(L_d, R_d)$ is less than $c$. Then, a transformation $T_c$ could be applied to the pair $(L_d, R_d)$ that reduces it into a (lower-dimensional) controllable pair and an uncontrollable pair. Since the pair $(L_d, R_d)$ cannot be decomposed into smaller blocks {as it is part of a finest SBD}, the {original} assumption is not satisfied. We thus conclude $c=h$.
}

\end{proof}
Henceforth $c$ denotes the dimension of the controllability subspace.

\begin{theorem}
The undriven  {system} \eqref{undriv} is composed of $N-c$ scalar equations.
\end{theorem}
\begin{proof}
Since (i) the Laplacian matrix $L$ is symmetric $L=L^{T}$ and (ii) the transformation matrix $T=SBD(L,R)$ is orthogonal $T^{-1}=(T)^{T}$, we have that,
\begin{equation}
\begin{array}{clc}
    (T^{-1}LT)^{T}=T^{T}L^{T}(T^{-1})^{T}=T^{-1}LT,
\end{array}
\end{equation}
which shows that the matrix $L_T=T^{-1}LT$ is symmetric. It also follows that the blocks $L_d$ and $L_{ud}$ are symmetric and diagonalizable.  By diagonalizing the matrix $L_{ud}$ (and taking into account the fact that the matrix $R_{ud}=0$), the undriven equation \eqref{undriv} can be decoupled into a number of scalar equations. 
\end{proof}
According to Theorems 1 and 2, the pinning control stability problem can be decoupled into two sets of equations: $c$ driven equations and $(N-c)$ undriven equations. Also, $c$ is equal to the rank of the controllability matrix of the pair $({L},{R})$.

Next we seek to find $P_{1}$ and $P_{2}$ such that the matrix $P$ commutes with the matrix ${L}$.
\begin{lemma}\label{lemmaP1P2}
The matrices $P_1$ and $P_2$ can be found by 
{choosing a vector in the nullspace} of the $s^2+\tau^2$-dimensional  matrix $K_1^T K_1+K_2^T K_2$ defined below.
\end{lemma}
\begin{proof}
By rewriting ${L}=\big(\begin{smallmatrix}
  L_{11} & L_{12}\\ 
  L_{21} & L_{22}
\end{smallmatrix}\big)$, the commutation equation $P{L}={L}P$ becomes
\begin{subequations}\label{11}
\begin{equation}\label{11a}
\begin{pmatrix}
  P_{1} & 0\\
  0 & P_{2}
\end{pmatrix} \begin{pmatrix}
  L_{11} & L_{12}\\ 
  L_{21} & L_{22}
\end{pmatrix}=\begin{pmatrix}
  L_{11} & L_{12}\\ 
  L_{21} & L_{22}
\end{pmatrix} \begin{pmatrix}
  P_{1} & 0\\
  0 & P_{2}
\end{pmatrix},
\end{equation}
\begin{equation}\label{13b}
    \begin{pmatrix}
  P_{1}L_{11} & P_{1}L_{12}\\
  P_{2}L_{21} & P_{2}L_{22}
\end{pmatrix}=\begin{pmatrix}
  L_{11}P_{1} & L_{12}P_{2}\\
  L_{21}P_{1} & L_{22}P_{2}
\end{pmatrix}.
\end{equation}
\label{Eq13}
\end{subequations}

Eq.\ \eqref{13b} corresponds to the following four equations,
\begin{subequations}\label{12}
\begin{equation}
P_{1}L_{11}=L_{11}P_{1}
\label{Eq14a}
\end{equation}
\begin{equation}
P_{2}L_{22}=L_{22}P_{2}
\label{Eq14b}
\end{equation}
\begin{equation}
P_{1}L_{12}=L_{12}P_{2}
\label{Eq14c}
\end{equation}
\begin{equation}
P_{2}L_{21}=L_{21}P_{1}
\label{Eq14d}
\end{equation}
\label{Eq14}
\end{subequations}
Define the operator $\text{vec} : \mathbb{R}^{n \times m} \mapsto \mathbb{R}^{nm}$ which takes as input a matrix and returns a vector consisting of the columns of the matrix stacked on top of each other.
Note that $P_1$ has dimension $s$ and $P_2$ has dimension $\tau = N-s$.
The four equations in Eq. \eqref{12} can be expressed as the following two linear systems of equations.
\begin{subequations}\label{eq:nullspaces}
  \begin{equation}
  \begin{array}{lcl}
  \left[ \begin{array}{cc}
    I_s \otimes L_{11} - L_{11} \otimes I_s & O_{s^2,\tau^2}\\
    O_{\tau^2,\tau^2} & I_\tau \otimes L_{22} - L_{22} \otimes I_\tau
  \end{array} \right] \left[ \begin{array}{c}
    \text{vec}(P_1) \\ \text{vec}(P_2)
  \end{array} \right] \\=
  \left[ \begin{array}{c}
    \boldsymbol{0}_{s^2} \\ \boldsymbol{0}_{\tau^2}
  \end{array} \right]
  \end{array}
  \end{equation} 
  \\
  \begin{equation}
      \left[ \begin{array}{cc}
        L_{12}^T \otimes I_s & - I_\tau \otimes L_{12}\\
        -I_s \otimes L_{21} & L_{21}^T \otimes I_\tau
      \end{array} \right] \left[ \begin{array}{c}
        \text{vec}(P_1) \\ \text{vec}(P_2)
      \end{array} \right] = \left[ \begin{array}{c}
        \boldsymbol{0}_{s\tau} \\ \boldsymbol{0}_{s\tau}
      \end{array} \right]
  \end{equation}
\end{subequations}
By inspecting Eq. \eqref{eq:nullspaces}, {the} vector $\boldsymbol{p} = \left[ \begin{array}{cc} \text{vec}(P_1)^T & \text{vec}(P_2)^T \end{array} \right]^T$ 
{must} lie in the intersection of the nullspaces of the two matrices that appear in Eq. \eqref{eq:nullspaces}.
Let $K_1 \in \mathbb{R}^{s^2+\tau^2 \times s^2+\tau^2}$ and $K_2 \in \mathbb{R}^{2s\tau \times s^2+\tau^2}$ be the matrices that appear in the first and second lines of Eq. \eqref{eq:nullspaces}, respectively.
{In the Supplementary Material Note 3, we show that $\mathcal{N}(K_1) \cap N(K_2) = \mathcal{N}(K_1^TK_1) \cap \mathcal{N}(K_2^T K_2) = \mathcal{N}(K_1^T K_1 + K_2^T K_2)$, {where the symbol $\mathcal{N}(M)$ indicates the null subspace of the matrix $M$.} Then, we create the vector $\boldsymbol{p} = \sum_{i=1}^{n_z} \alpha_i v_i$ where $n_z$ is dimension of the nullspace of the matrix $S = K_1^T K_1 + K_2^T K_2$ and $v_i$ are a basis for the nullspace of $S$, that is, eigenvectors corresponding to its zero eigenvalues.}
With $\boldsymbol{p}$, form the two matrices $P_1$ and $P_2$ and construct the full matrix $P$ that appears in Eq. \eqref{11a}.
Additional details of the derivation of Eq. \eqref{eq:nullspaces} can be found in Appendix \ref{Ap1}.\\
\end{proof}
\begin{remark}
 The procedure outlined to find the block diagonal matrix $P$ is a specialization of the method presented in \cite{SBD2} to simultaneously block diagonalize two matrices when one of the matrices has the form of $R$ in Eq. \eqref{Btilde}.
Using the method of \cite{SBD2} directly on $L$ and $R$ requires finding a vector in the nullspace of a square matrix of dimensions $(s+\tau)^2$ while the method outlined in Lemma \ref{lemmaP1P2} requires finding a vector in the nullspace of a square matrix of dimension $s^2+\tau^2$.
As $s^2 + \tau^2 < (s+\tau)^2$ for $s,\tau \geq 1$, the method presented here is more efficient than the method of \cite{SBD2} by exploiting the structure of $R$.
\end{remark}
To illustrate the procedure described in this section, consider the network with inputs {with $N=5$ nodes} shown in Fig.\ \ref{new} (a).
\begin{figure}[H]
\centering
  \includegraphics[scale=.25]{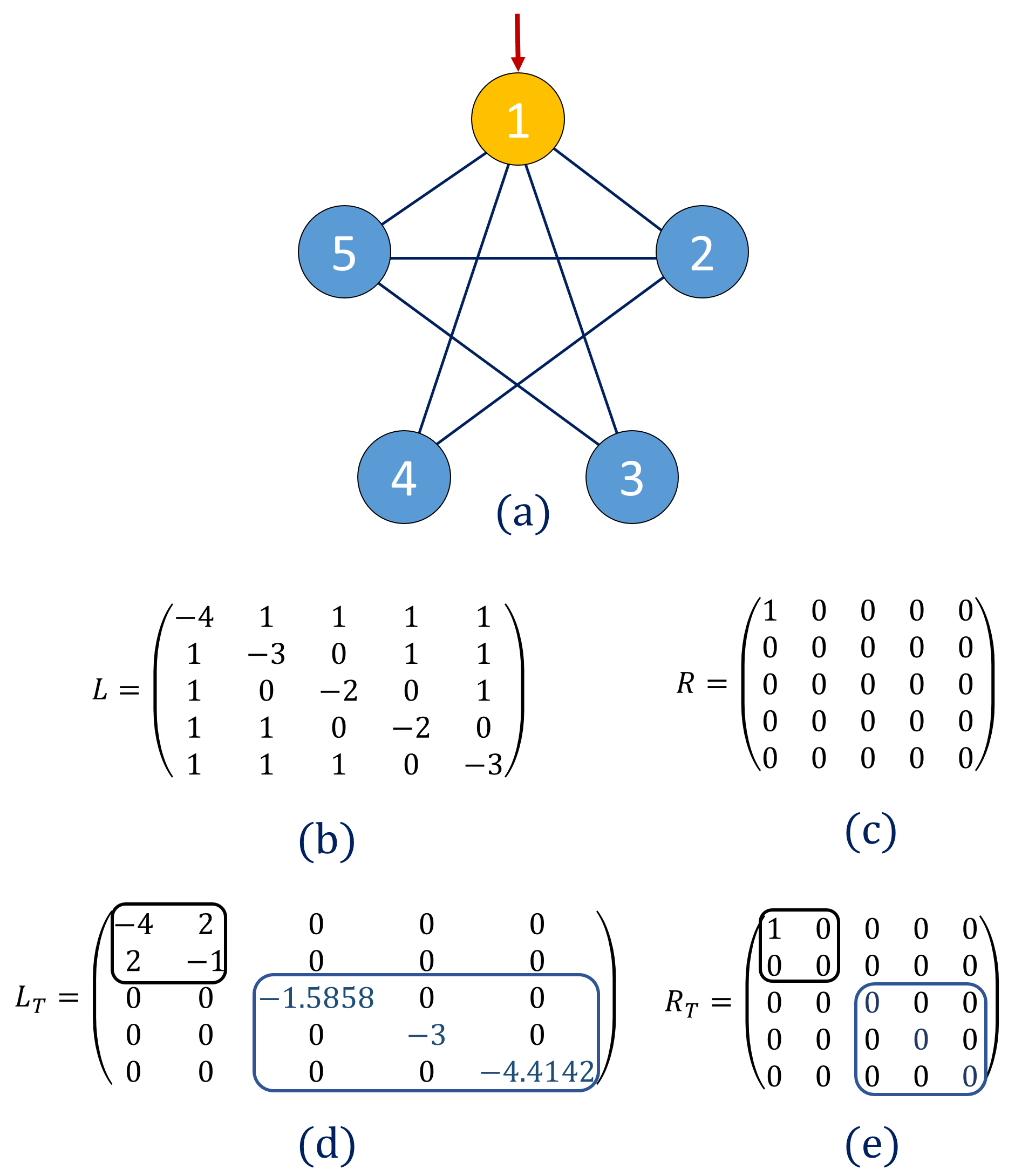}
\caption{(a)  {A} $5$-dimensional network with inputs. Nodes are color-coded according to the equitable clusters (see Definition \ref{EC}) to which they belong. The pinned node (node number $1$) in yellow is in a cluster, while all the remaining nodes (node $2$ to $5$) in blue are in another cluster. The effect of the pinning control action is graphically shown as a red arrow pointing at node $1$. {Matrices} (b) and (c) show the pair $(L, R)$ corresponding to this network
{and matrices} (d) and (e) show the pair $(L_T, R_T)$ obtained by application of the SBD transformation $T$. {The pair $(L_T, R_T)$  contains driven blocks (inside black boxes) and undriven blocks (inside blue boxes.)}
}
\label{new}
\end{figure}

To compute the transformation matrix $T$ for the pair $(L, R)$ shown in Fig.\ \ref{new} (b,c), we first find the block diagonal matrix $P$ using Lemma \ref{lemma P} and Lemma \ref{lemmaP1P2} and obtain:
\begin{equation}
   P= \begin{pmatrix}
      0.6828 & 0 & 0 & 0 & 0\\
      0 & 0.0036 & 0.7265 & 0.0801 & -0.1274\\
      0 & 0.7265 & -0.6427 & 0.5189 & 0.0801\\
      0 & 0.0801 & 0.5189 & -0.6427 & 0.7265\\
      0 & -0.1274 & 0.0801 & 0.7265 & 0.0036
    \end{pmatrix}
\end{equation}
Now by constructing the transformation matrix $T$ from the eigenvectors of the matrix $P$ and following  Eq. \eqref{eq:directsumLR} we obtain the pair of matrices $(L_T, R_T)$ which are shown in Fig. \ref{new} (d,e): The driven part (black blocks in Fig. \ref{new}) has dimension $2$ because the controllability subspace of the pair ($L,R$) has dimension $c=2$.
\begin{figure}[b!]
     \centering
     (a)\subfigure{\includegraphics[width=0.8\linewidth]{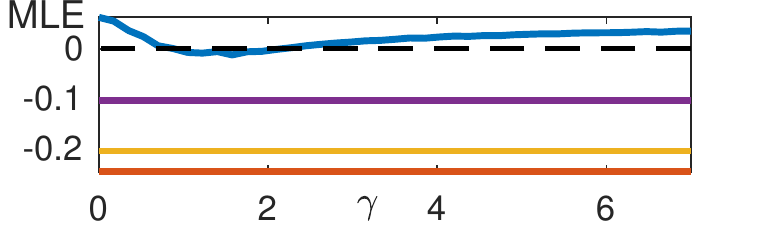}}
     (b)\subfigure{\includegraphics[width=0.8\linewidth]{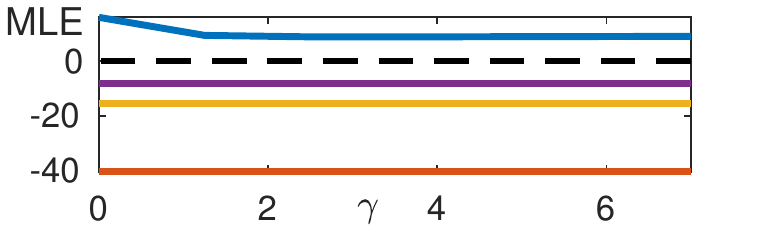}}
     (c)\subfigure{\includegraphics[width=0.8\linewidth]{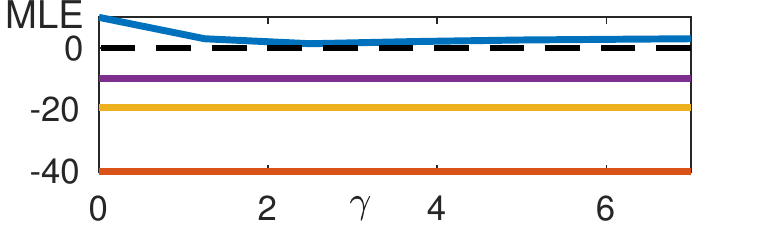}}
     (d)\subfigure{\includegraphics[width=0.8\linewidth]{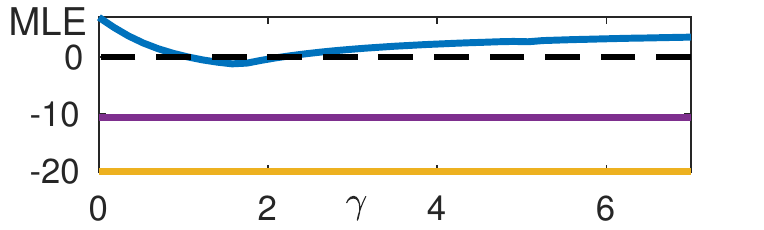}}
\caption{MLEs related to ($\hat{L}_1,\hat{R}_1$) (blue line), ($\hat{L}_2,\hat{R}_2$) (red), ($\hat{L}_3,\hat{R}_3$) (yellow), ($\hat{L}_4,\hat{R}_4$) (violet), for the example in Fig. \ref{new} when the target trajectory is either the stable chaotic attractor exhibited by a single oscillator (a) or a 1-cycle UPO  (b), or a 4-cycle UPO (c), or a  8-cycle UPO (d).}
     \label{fig:res_ex1}
 \end{figure}

{As an example to demonstrate the use of the MLE to determine the stability of the system, we make each of the nodes in the network in Fig.\ \ref{new} a R\"{o}ssler oscillator so that $m=3$ where the state vector of each node is governed by the function}
\begin{equation}\label{eq:ross}
 	\bF(\bx_i) = \begin{bmatrix}
 	- y_i - z_i \\
 	 x_i + a y_i\\
 	b + (x_i - c) z_i\\
 	\end{bmatrix}
 \end{equation} 
 where $a = 0.1$, $b = 0.1$ and $c = 15$. With these parameters the oscillator evolves within a stable chaotic attractor. The network {node-to-node} coupling function is $\bG(\bx_i) = [x_i, 0, 0]^T$, whereas the chosen {input-to-node coupling} function is $\bH(\bx_i) = [0, y_i, 0]^T$.
Fig.\ \ref{fig:res_ex1} shows the MLEs associated to Eqs. \eqref{driv} and \eqref{undriv} plotted vs $\gamma$ for the network of Fig. \ref{new}.

Fig. \ref{fig:res_ex1}(a) shows that the nodes of the network synchronize on the target trajectory for $\gamma \in [0.9,  2.3]$. In particular, it is apparent that the driven equations \eqref{driv} depend on $\gamma$, whereas the undriven equations \eqref{undriv} are independent. 
\begin{figure}[t!]
    \centering
     (a) \subfigure{\includegraphics[width=0.92\linewidth]{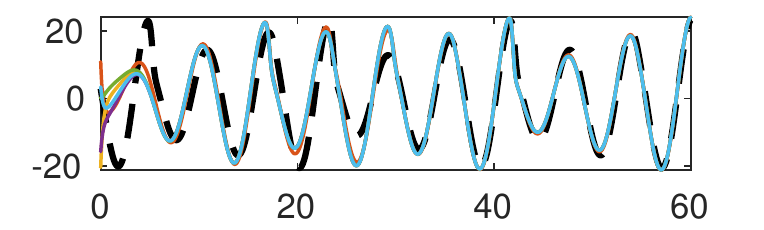}}
     (b) \subfigure{\includegraphics[width=0.92\linewidth]{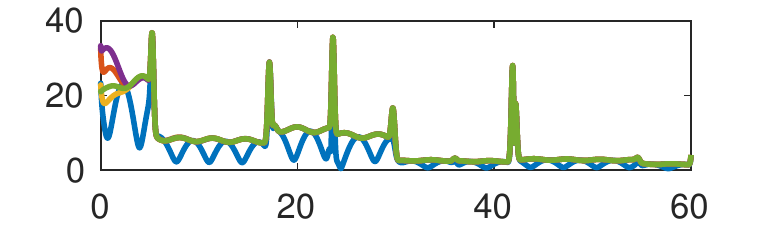}}
    \caption{Time plots for $\gamma = 1.5$. (a) Evolution of the state variables $x_i$ (solid lines, $i=1$ blue, $i=2$ red, $i=3$ yellow, $i=4$ purple, $i=5$ green) and of the first component of the target trajectory (dashed line). (b) Evolution of $E_i(t) = ||\bx_i - \bx_t ||_2$; color code as in panel (a).}
    \label{fig:time_evolution_chatoic}
\end{figure}
As a harder benchmark, {we assign an unstable periodic orbit (UPO) embedded in the attractor as the target trajectory}
which can be selected by fine-tuning the initial condition of Eq. \eqref{Eq3}.\footnote{We computed the UPO by using the continuation tool MATcont \cite{dhooge:2003,dhooge2006matcont}.}
According to the theory \cite{hunt:1996}, we expect node trajectories to converge to a target trajectory corresponding to a UPO only for a relatively high period: while Figs. \ref{fig:res_ex1}(b,c) show that it is not possible to have convergence to an unstable 1-cycle or 4-cycle due to the presence of a positive MLE for any $\gamma$ value, Fig. \ref{fig:res_ex1}(d) shows that the nodes of the network synchronize on the target trajectory (unstable 8-cycle) for $\gamma \in [1 \quad 2.2]$. This is confirmed by Fig. \ref{fig:time_evolution_chatoic}a, which shows the time evolution of the state variables $x_i$ (solid line) and of the corresponding component (dashed line) of the target trajectory $\textbf{x}_t$ (unstable 8-cycle) for $\gamma = 1.5$ and for random initial conditions. Fig.\ \ref{fig:time_evolution_chatoic}b shows the time evolution of the Euclidean norm $E_i$  of the difference $\bx_i - \bx_t$, from which we see that $\bx_i$ converges to $\bx_t$ for large time.
\section{Dimensionality reduction  through an alternative transformation $\hat{T}$} \label{results}
In what follows we discuss two transformations, each one of which results in a simultaneous block diagonalization
(although not necessarily finest) of the pair $(L,R)$: one is
the transformation $T_c$ which decouples the problem \eqref{Eq11} into a controllable equation and an uncontrollable equation;
the other one is the transformation associated with the quotient network \cite{sanchez2020exploiting, golubitsky2012singularities, golubitsky2003symmetry}, which in turn decouples the problem \eqref{Eq11} into two sets of equations, which will be referred to as quotient and redundant. 
By applying both transformations to the pair $(L,R)$, we obtain four independent equations: quotient controllable, quotient uncontrollable, redundant controllable, redundant uncontrollable. {Interestingly,} this decomposition into 4 independent {sets of} equations {has the same structure as} the decomposition produced by a finest SBD.

\subsection{Matrix $T_c$}\label{sec:matrix_tc}
From definition \ref{standard}, it follows that there is a transformation $T_c(L,R)$, which divides the system given by the pair $(L,R)$ in a controllable subsystem and an uncontrollable subsystem, {see also \cite{martini:2010}},
\begin{equation} \label{T1}
T_c^{-1} L T_c=\begin{pmatrix}
  L_C & 0\\
  0 & L_{U}
\end{pmatrix}
\quad \quad
T_c^{-1} R=\begin{pmatrix}
  R_C & 0\\
  0  & 0
\end{pmatrix},
\end{equation}
where the controllable subsystem is given by the $c$-dimensional pair $(L_C,R_C)$ and the uncontrollable subsystem is given by the $(N-c)$-dimensional pair $(L_{U},0)$.

\begin{remark}\label{rem:build_tc}
The first $c$ linearly-independent column vectors of the matrix $T_c$ can be chosen as an orthonormal basis for the range of the Kalman controllability matrix $K$. The remaining $N-c$ columns of $T_c$ are chosen such that $T_c$ is a basis for $\mathbb{R}^{N}$.
\end{remark}
\color{black}

\begin{lemma}
The matrix $T_c^{-1} R T_c$  {has} the same block structure as the matrix $T_c^{-1} R$ in Eq.\ \eqref{T1}.
\end{lemma}
\begin{proof}
Due to the fact that matrix $L$ is symmetric, the transformation $T_c$ can be constructed to be an orthogonal matrix $T_c^{-1}=T_c^{T}$ which can be splitted into four sub-matrices as:
\begin{equation}
        T_c=
    \begin{pmatrix}
        T_{C1} & T_{C2}\\
        T_{C3} & T_{C4}
    \end{pmatrix}
\end{equation}
where the two square sub-matrices $T_{C1} \in \mathbb{R}^{s\times s}$ and $T_{C4} \in \mathbb{R}^{N-s\times N-s}$ and the two rectangular sub-matrices $T_{C2} \in \mathbb{R}^{s\times N-s}$ and $T_{C3} \in \mathbb{R}^{N-s\times s}$. Now according to Remark \ref{rem:build_tc} and by considering that (i) the matrix $R$~\eqref{Btilde} is symmetric and (ii) $T_c^{-1} R$ has the special block-diagonal form of Eq.\ \eqref{T1}, we have
\begin{equation}
    T_c^{T}R=
    \begin{pmatrix}
        T_{C1}^{T} & T_{C3}^{T}\\
        T_{C2}^{T} & T_{C4}^{T}
    \end{pmatrix}
    \begin{pmatrix}
        R_1 & 0\\
        0 & 0
    \end{pmatrix}
    =
    \begin{pmatrix}
        T_{C1}^{T}R_1 & 0\\
        T_{C2}^{T}R_1 & 0\\
    \end{pmatrix}.
\end{equation}
We also note that based on definition~\ref{standard}, $T_{C2}^{T} R_1=0$.

From Eq.\ \eqref{Btilde}, we know that $R_1$ is equal to the identity matrix $I_s$, which implies that $T_{C2}^{T}=0$. Then for the matrix $T_c^{T} R T_c$ we have:
\begin{equation}\label{Tcnew}
    T_c^T R T_c=
    \begin{pmatrix}
  R_C & 0\\
  0  & 0
\end{pmatrix}
    \begin{pmatrix}
        T_{C1} & 0\\
        T_{C3} & T_{C4}
    \end{pmatrix}=
        \begin{pmatrix}
  R_C T_{C1} & 0\\
  0  & 0
\end{pmatrix},
\end{equation}
which implies that the matrix $T_c^{T} R T_c$  {has} the same block structure as the matrix $T_c^{T} R$ in Eq.\ \eqref{T1}.

\end{proof}
\color{black}

It should be noted that the transformation $T_c$ is not unique.
In this paper we build $T_c = [U_c \quad X_{N-c}]$ as follows.
We define the singular value decomposition (SVD) of $K$ as $K = U \Sigma V^T$. $U_c$ contains the first $c$ columns of $U$ and $X_{N-c}$ contains the $N-c$ eigenvectors of $L$ that do not belong to the controllability subspace.
\color{black}

\subsection{Matrix $T_q$}\label{Sec.Tq}
We now introduce another transformation, which we refer to as the quotient transformation, which we show also leads to simultaneously block-diagonalizing the pair $(L,R)$.

\begin{definition}\textbf{Equitable cluster partition.} \label{EC}
Given the network with inputs (defined in \eqref{Lr}) described by the pair $({L},R)$, one can partition the set of the network nodes $\mathcal V$ into subsets, i.e., equitable clusters, $\mathcal{C}_1,\mathcal{C}_2,..,\mathcal{C}_M$, $\cup_{i=1}^M \mathcal{C}_i=\mathcal{V}$, $\mathcal{C}_i \cap \mathcal{C}_j=\emptyset$ for $i \neq j$, where
\begin{equation} \label{ecp}
  \sum_{k \in \mathcal{C}_{\ell}} \tilde{L}_{ik} = \sum_{k \in \mathcal{C}_{\ell}} \tilde{L}_{jk}, \quad \begin{aligned} \forall i,j &\in \mathcal{C}_k\\ \forall \mathcal{C}_k,\mathcal{C}_{\ell} &\subset \mathcal{V}. \end{aligned}
\end{equation}
\begin{remark} \label{remark pinned}
Each of the equitable clusters can either {contain} only nodes with external inputs (pinned nodes) or only nodes without external inputs (non-pinned nodes). No cluster
{contains} both pinned and non-pinned nodes. It follows that either $\mathcal{C}_k \subset \mathcal{V}_P$ or $\mathcal{C}_k \subset \mathcal{V}_{NP}$, $k=1,...,M$.
\end{remark}
\end{definition} 

In order to find the equitable clusters, we apply the algorithm developed by Belykh and Hasler \cite{belykh2011} to the extended network with Laplacian matrix $\tilde{L}$. The extended network can always be partitioned into $M+1$ equitable clusters in which one of the clusters necessarily corresponds 
{to} the only source node. In what follows we are going to neglect this one cluster and focus on the remaining $M$ equitable clusters $\mathcal{C}_1,\mathcal{C}_2,\ldots,\mathcal{C}_M$. 
{Now we can define the $N\times M$ equitable cluster indicator matrix $E$, where $N$ is the number of nodes and $M$ is the number of clusters: $e_{ij}=1$ if node $i$ is in cluster $\mathcal{C}_j$ and $0$ otherwise for $i=1,\ldots,N$ and $j=1,\ldots,M$.}

\begin{definition}\textbf{Quotient Graph.}\label{defquotient}
Given a network with inputs, represented by the matrix pair $(L, R)$, and the equitable clusters, the quotient network $(L_Q, R_Q)$ can be defined as a network where all the nodes in each equitable cluster are mapped to a single quotient node, 
\begin{subequations}
\begin{align}
L_Q &= (E^T E)^{-1/2} E^T L E (E^TE)^{-1/2}\\
R_Q &= (E^T E)^{-1/2} E^T R E (E^TE)^{-1/2} .
\end{align}
\end{subequations}
\end{definition}

We introduce the orthogonal transformation matrix ${T}_q$ whose first $M$ rows are $T_q^{E}=(E^T E)^{-1/2} E^T$ and whose remaining $N-M$ rows are orthogonal to the first ones. We call $T_q^r$ the matrix composed of the last $N-M$ rows of the matrix ${T}_q$. We take the matrix $T_q^r$ to have a particular structure, where each one of its rows is associated with one cluster, in the sense that all of the entries in that row not corresponding to the nodes in that cluster are equal to zero. Several such matrices have been proposed in the literature \cite{NC,schaub2016graph,cho2017stable}.

\color{black}
\begin{theorem} \label{lemmaquotient}
 Application of the `quotient' transformation ${T}_q$ to the pair $L,R$ yields,
\begin{subequations}
\begin{align}
{T}_q L {T}_q^{-1} &=L_Q \oplus L_O \label{a}\\
{T}_q R {T}_q^{-1} &=R_Q \oplus R_O \label{b},
\end{align}
\end{subequations}
where the square quotient matrices $L_Q,R_Q$ have size $M \times M$ and the square `redundant' matrices $L_O,R_O$ have size $(N-M) \times (N-M)$.
\end{theorem}
\begin{proof}
For succinctness we omit the proof of \eqref{a}. We just note that this proof follows {directly} from Refs.\ \cite{NC,schaub2016graph,cho2017stable}. 

{To prove Eq.\ \eqref{b}, we}
know from Eq.\ \eqref{Btilde} that the matrix $R$ is diagonal with the first $s$ entries being one and the remaining $N-s$ entries being zero. Therefore, the matrix $R$ has $s$  eigenvalues equal one and $N-s$ eigenvalues equal zero. Then, we can write the matrix $R=V_r D_r V_{r}^{T}$, where the matrix {$D_r=D1\oplus D{0}$} such that $D_1$ is the $s$-dimensional identity matrix and $D_{0}=0_{N-s}$. The matrix $V_r$ can be written in the form,
\begin{equation}
V_r=\begin{pmatrix}
({V}_1)_{s\times s}  &  \mathbf{0}_{s\times (N-s)}\\
\mathbf{0}_{(N-s)\times s} & ({V}_0)_{(N-s)\times (N-s)}
\end{pmatrix}
\end{equation}
due to the property that the eigenvectors associated with the $1$ eigenvalue ($0$ eigenvalue) have all of their last $N-s$ (first $s$) entries equal to zero.
Now we can write,
\begin{equation}
{T}_q R {T}_q^{-1} ={T}_q V_r D_r V_{r}^{T} {T}_q^{T}=({T}_q V_r) D_r ({T}_q V_r)^{T}.
\label{quotient3}
\end{equation}
It follows that
\begin{equation}
{T}_q R {T}_q^{-1}=\begin{pmatrix}
T_q^{E_{p}} (T_q^{E_{p}})^T & T_q^{E_{p}} (T_q^{r_{p}})^T\\
T_q^{r_{p}} (T_q^{E{p}})^T & T_q^{r_{p}} (T_q^{r_{p}})^T\\
\end{pmatrix},
\end{equation}
where $T_q^{E_{p}}$ is the $M\times s$ sub-matrix that shows the restriction of the first $m$ rows of $T_q$ corresponding to the clusters that contain pinned nodes. $T_q^{r_{p}}$ is the $(N-M)\times s$ sub-matrix that shows the restriction of the remaining $N-M$ rows of $T_q$ corresponding to the clusters that contain non-pinned nodes. Now considering the particular structure of $T_q$ 
and the characteristic of the equitable clusters of a network with inputs (Remark \ref{remark pinned}), we have $T_q^{E_{p}} (T_q^{r_{p}})^T=0$. Therefore, it can be concluded that ${T}_q R {T}_q^{-1} =R_Q \oplus R_O$ where $ R_Q=T_q^{E_{p}} (T_q^{E_{p}})^T$ and $R_O=T_q^{r_{p}} (T_q^{r_{p}})^T$.
\end{proof}
\color{black}
We conclude that application of the transformation $T_q$ to the pair $(L,R)$, $T_q^{-1} L T_q$ and $T_q^{-1} R T_q$,
simultaneously block diagonalizes $(L,R)$ into an $M$-dimensional pair $(L_Q,R_Q)$ and an $(N-M)$-dimensional pair $L_O,R_O$.

\begin{remark} \label{R}
 The results of Theorem 3 apply to any equitable cluster partition associated with the pair {$(L ,R)$}, not just the equitable cluster partition with the smallest number of clusters returned by the algorithm developed by Belykh and Hasler \cite{belykh2011}.
\end{remark}

\subsection{Transformation $\hat{T}$}\label{Sec.Tc}
Next we discuss {the} application of both transformations $T_c$ and $T_q$ to the pair $(L,R)$. While there are different ways in which this can be done, we found the most straightforward approach is to first apply $T_q$ and decouple the pair $(L,R)$ into  the pairs $(L_Q,R_Q)$ and $(L_O,R_O)$ and then apply equivalent transformations to the pairs $(L_Q,R_Q)$ and $(L_O,R_O)$ into their controllable and uncontrollable parts. We call $\hat{T}$ the resulting transformation. 

\begin{lemma} \label{Lemma hat T}
The pair $(L,R)$ is decoupled by application of the resulting transformation $\hat{T}$  
into four blocks,
\begin{equation}\label{Eq:hatT}
\begin{array}{clc}
\hat{T} L \hat{T}^{-1}=\begin{pmatrix}
  L_{qc} & 0 & 0 &0\\
  0 & L_{rc} & 0 &0\\
  0 & 0 & L_{qu} &0\\
  0 & 0 & 0 & L_{ru}
\end{pmatrix}
\\\\
\hat{T} R \hat{T}^{-1}=\begin{pmatrix}
   R_{qc} & 0 & 0 &0\\
  0 & R_{rc} & 0 &0\\
  0 & 0 & 0 &0\\
  0 & 0 & 0 & 0
\end{pmatrix},
\end{array}
\end{equation}
where $(L_{qc},R_{qc})$ is the controllable quotient pair,  $(L_{rc},R_{rc})$ is the controllable redundant pair, $(L_{qu},0)$ is the uncontrollable quotient pair, and $(L_{ru},0)$ is the uncontrollable redundant pair. The size of the pair $(L_{qc},R_{qc})$ is equal to the dimension $u$ of the controllable subspace of the quotient pair $(L_Q,R_Q)$, the sum of the sizes of the two pairs $(L_{qu},0)$ and $(L_{ru},0)$ is equal to the dimension $N-c$ of the {uncontrollable }subspace of the pair $(L,R)$. The size of the pair $(L_{rc},R_{rc})$ is equal to $c-u$. 
\end{lemma}
\begin{proof}
From Theorem \ref{lemmaquotient} we see that application of the transformation ${T}_q$ to the pair $(L, R)$ {decouples} the pair into a quotient pair $(L_Q, R_Q)$ and a redundant pair $(L_O, R_O)$. Application of the equivalent transformation to the pair $(L_Q, R_Q)$ will necessarily decouple the pair $(L_Q, R_Q)$ in its controllable and {uncontrollable } pairs and  application of the equivalent transformation to the pair $(L_O, R_O)$ will necessarily decouple the pair $(L_O, R_O)$ in its controllable and
{uncontrollable } pairs.

\end{proof}
Table \ref{Table2} summarizes the application of transformation {$\hat{T}$} that decouples the pair $(L, R)$ into four block pairs, {which in what follows we will refer to as the qc pair (quotient controllable), the qu pair (quotient uncontrollable), the rc pair (redundant controllable), and the ru pair (redundant uncontrollable.)}
\begin{table}[h!]
\caption{The four decoupled blocks and their dimensions after applying transformation $\hat{T}$ {to} the pair ($L, R$). The first two rows show the controllable pairs, $qc$ and $rc$, with dimensions $u$ and $c-u$, respectively. The second two rows indicate the uncontrollable pairs, $qu$ and $ru$. The sum of the size of the uncontrollable pairs is equal to $N-c$.
}
\centering
  \includegraphics[scale=.32]{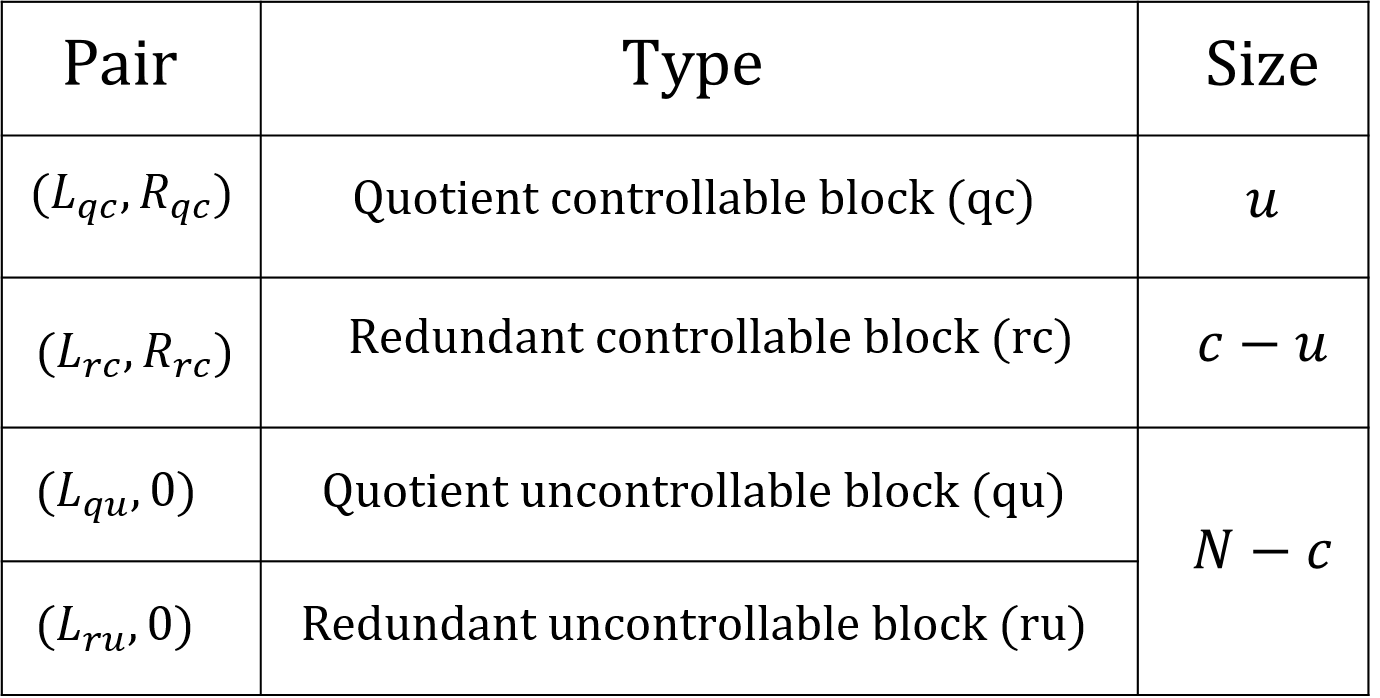}
\label{Table2}
\end{table}

\begin{theorem} \label{Theorem4}
    The transformation $\hat{T}$ leads to a finest SBD of the pair (L,R).
\end{theorem}
\begin{proof}
    The proof is based on the concept of common invariant subspaces \cite{bischer2021simultaneous,arapura2004common}.
    As stated above, we first apply $T_q$ and decouple the pair $(L,R)$ into  the pairs $(L_Q,R_Q)$ and $(L_O,R_O)$.
We initially focus on the pair $L_Q = T_q^E L (T_q^E)^T,R_Q  = T_q^E R (T_q^E)^T$ and on the corresponding column subspaces $\mathcal{L}_Q$ and $\mathcal{R}_Q$. We compute the matrix $T_{c}$ as described at the end of Sec. \ref{sec:matrix_tc} and we obtain: 
\begin{equation}
    T_{c}^{-1} T_q^E L  (T_q^E)^T T_{c} = T_{c}^{-1} L_Q T_{c} = \begin{bmatrix}L_{qc} & 0 \\ 0 & L_{qu} \end{bmatrix}
\end{equation}
\begin{equation}
    T_{c}^{-1} T_q^E R  (T_q^E)^T T_{c} = T_{c}^{-1} R_Q T_{c} = \begin{bmatrix}R_{qc} & 0 \\ 0 &  0 \end{bmatrix}
\end{equation}

{The column subspaces of the matrices $L_{qc}$, $R_{qc}$ and $L_{qu}$ are $\mathcal{L}_{qc}$, $\mathcal{R}_{qc}$ and $\mathcal{L}_{q}$, respectively}.
 The linearly independent columns of $B_L = L  (T_q^E)^T$ form a basis for the span of $L_Q$. Analogously, the linear independent columns of $B_R = R  (T_q^{{E}})^T$ form a basis for the span of $R_Q$. Therefore, $\hat{B} = B_R^T B_L$ contains the $p$ linearly independent columns of the product $L_Q R_Q$.\\

We compute the SVD $\hat{B} = \hat{U} \hat{\Sigma} \hat{V}^T$. The matrix $\hat{U}_p$  contains the first $p$ columns of $\hat{U}$, which form a basis of $\mathcal{R}_Q \cap \mathcal{L}_Q$ \cite{van1996matrix}.

By using the basis $U_p$, the intersection space $\mathcal{R}_Q \cap \mathcal{L}_Q$ can be expressed as the direct sum of minimal-dimension subspaces\cite{galantai2008subspaces}. \\

The controllability matrix of the pair $(L_Q,R_Q)$ is $K = [R_Q, L_Q R_Q, \dots, L_Q^{N-1} R_Q ]$ and certainly contains  the columns of $\hat{B}$. Therefore, the columns of $U_p$ are a subset of those of $U_c$ (see. Remark \ref{rem:build_tc} and following text).

The matrix $T_c$ decomposes the vector space $\mathcal{R}_Q + \mathcal{L}_Q$ into the following direct sum of subspaces:
\begin{equation}
\resizebox{0.88\hsize}{!}{$
\begin{array}{cc}
    \mathcal{R}_Q + \mathcal{L}_Q = (\mathcal{L}_{qc} + \mathcal{R}_{qc}) \oplus \mathcal{L}_{qu} = \\\\
    \left[(\mathcal{R}_{qc} \cap \mathcal{L}_{qc}) \oplus (\mathcal{R}_{qc}-\mathcal{R}_{qc} \cap \mathcal{L}_{qc} ) \oplus (\mathcal{L}_{qc}-\mathcal{R}_{qc} \cap \mathcal{L}_{qc} )\right] \oplus \mathcal{L}_{qu}
    \end{array}
    $}
\end{equation}
As stated above, $(\mathcal{R}_{qc} \cap \mathcal{L}_{qc})$ is the direct sum of minimal-dimension subspaces because it is represented through the basis formed by the columns of $U_c$. The vector space $(\mathcal{R}_{qc}-\mathcal{R}_{qc} \cap \mathcal{L}_{qc} ) \oplus (\mathcal{L}_{qc}-\mathcal{R}_{qc} \cap \mathcal{L}_{qc} )$ cannot be decomposed in the direct sum of smaller common invariant subspaces because it is the direct sum of two subspaces with empty intersection.

 The vector space $\mathcal{L}_{qu}$ is the direct sum of $1$-dimensional subspaces, because it is represented through the basis formed by the columns of $X_{N-c}$, which contain the eigenvectors of $L_Q$ that belong to the non controllable subspace.

Since the transformation performed by $T_c$ decomposes the vector space  $\mathcal{L}_Q+\mathcal{R}_Q$ into the direct sum of minimal-dimensional subspaces, it produces a finest SBD of the transformed matrices $T_{c}^{-1} L_Q T_{c}$ e $T_{c}^{-1} R_Q T_{c}$. \\

Similar reasoning can be applied, \textit{mutatis mutandis}, to the pair $(L_o,R_o)$. Therefore, the transformtation $\hat{T}$ produces a finest SBD of the pair $(L,R)$.
\end{proof}


To demonstrate the application of the transformation $\hat{T}$ {to a} pinning control problem, consider the network shown in Fig.\ \ref{new} (a). The network has $M=2$ equitable clusters, one comprising only the pinned node (in yellow in the figure) and the other one comprising all the remaining nodes (in blue in the figure). The $M=2$-dimensional quotient network is controllable, therefore it cannot be reduced further. By computing the transformation matrix $\hat{T}$ and following Lemma \ref{Lemma hat T} and Eq.\ \eqref{Eq:hatT}, we have:
\begin{subequations}
\begin{equation}
    \hat{T}L\hat{T}^{-1}=\begin{pmatrix}
-1.8937 & 2.4254 & 0 & 0 & 0 \\ 
2.4254 & -3.1063 & 0 & 0 & 0 \\ 
0 & 0 & -1.5858 & 0 & 0 \\ 
0 & 0 & 0 & -3 & 0 \\ 
0 & 0 & 0 & 0 & -4.4142 
    \end{pmatrix}
    \end{equation}
    \begin{equation}
    \hat{T} R \hat{T}^{-1}=\begin{pmatrix}
-4.9981 & 0.0970 & 0 & 0 & 0 \\ 
0.0970 & -0.0019 & 0 & 0 & 0 \\ 
0 & 0 & 0 & 0 & 0 \\ 
0 & 0 & 0 & 0 & 0 \\ 
0 & 0 & 0 & 0 & 0
\end{pmatrix}
\end{equation}
\end{subequations}
{By applying the transformation $\hat{T}$ to the network of Fig.\ \ref{new}, the pair $(L, R)$ is decoupled into a} $2$-dimensional qc block and three scalar ru blocks. By further diagonalizing the qc block of the matrix $\hat{T}R\hat{T}^{-1}$ and the ru block of the matrix $\hat{T}L\hat{T}^{-1}$, we obtain {the same transformation of the matrices $L$ and $R$ shown in Fig.\ \ref{new} (d,e).
}

We also considered application of the theory to randomly generated networks.
First, we constructed matrices $L$ corresponding to $500$ Erd\H{o}s-R\`{e}nyi random networks with $N=50$ nodes and connection probability $p=0.05$. Second, to construct the matrix $R$, we randomly selected the pinned nodes for the case that the number $s$ was varied from $1$ to $49$. Finally, we applied both the transformations $T$ and $\hat{T}$  to the constructed pair $(L, R)$, and we compared the number and dimensions of the resulting blocks, which yields a total of $24500$ numerical experiments. We repeated the process for all the combinations of $L$ and $R$. We found that in all the cases considered, the number and dimensions of the blocks produced by the SBD transformation {$T$} matched the number and dimensions of the blocks produces by $\hat{T}$ transformation, which is in accordance with the result of Theorem \ref{Theorem4} that $\hat{T}$ provides a finest SBD.

{We finally comment on the complexity of computing the two transformations $T$ and $\hat{T}$. The computational complexity of calculating the SBD transformation $T$ is 
$O(N^3)$ following \cite{zhang2020}. The computational complexity of calculating the transformation $\hat{T}= T_q T_c$ 
is 
$O(N^3)$ \cite{kalman1963mathematical,blondel2000survey,cline2006computatio}. We thus conclude that 
the computational complexity of the two calculations is of the same order.}


\color{black}

\section{Pinned node selection and multiple driven blocks}

As stated in Sec.\ \ref{ss:d&ud}, the dimension of the driven pair $(L_d,R_d)$ coincides with the rank of the controllability matrix of the pair $({L},{R})$. This indicates that both the topology of the network and the number and choice of the pinned nodes affect the dimension of the driven pair. In general, it is possible that the driven pair may be formed of $1 \leq w \leq s$  independent driven pairs $(L_d^1,R_d^1)$, $(L_d^2,R_d^2)$,..., $(L_d^w,R_d^w)$, where $L_d \oplus_{i=1}^w L_d^i$ and $R_d \oplus_{i=1}^w R_d^i$.
If $s=1$ then there is only one driven pair. However, the case $s>1$ in which there is more than one pinned node needs further consideration.
In what follows we show that the choice of the pinned nodes in relation to the network topology determines $w$. 

\begin{definition}\textbf{Pinned nodes symmetry.}\label{pns}
A pinned node symmetry (PNS) is a symmetry of the pair $(L, R)$ \cite{della2020}, i.e. a permutation matrix $\Pi$ with the following two properties: (1) it commutes with both matrices $L$ and $R$, i.e, $L \Pi=\Pi L$, $R \Pi=\Pi R$ and (2) it swaps two or more pinned nodes.
\end{definition}
Because the permutation matrix $\Pi$ commutes with both $L$ and $R$, it provides a particular solution to the problem of finding the matrix $P$ described in Sec.\ \ref{Sec:P}. Hence, the matrix $T_{\Pi}$ formed of the eigenvectors of the matrix $\Pi$ provides an SBD transformation (though not necessarily finest.) 

\begin{definition}\textbf{Eigenvectors of a permutation matrix \cite{garcia2015}.}\label{defeppns}
Each permutation matrix $\Pi$ partitions the set of the network nodes into $f$ disjoint cycles, $\mathcal{K}_1,\mathcal{K}_2,...,\mathcal{K}_f$ of length $l_1,l_2,...,l_f$, respectively, $\sum_{i=1}^f l_i=N$.  For instance, a cycle of length $1$ corresponds to a node that gets mapped back to itself, a cycle of length $2$ corresponds to two nodes that get swapped with one another and so on.
Each eigenvector of the permutation matrix is associated with just one cycle, meaning that the entries of that eigenvector are nonzero only for entries that correspond to nodes in the cycle. The number of eigenvectors associated with each cycle is equal to the length of the cycle. As a result, the matrix $T_\Pi$ can be written in block diagonal form, $T_\Pi= \bigoplus_{i=1}^f T_\Pi^i$, where each square block $T_\Pi^i$ has size $l_i$. For each cycle, there is one eigenvector whose entries associated with the nodes of the cycle are all ones and the remaining entries are all zeros. The remaining eigenvectors associated with each cycle have the same structure and are all orthogonal to one another.
\end{definition}

\begin{lemma}\label{lemma7}
 Choosing two or more pinned nodes that are swapped by a PNS results into independent driven pairs for each one of these pinned nodes.
\end{lemma}
\begin{proof}
The proof we provide is for the case that the PNS involves two pinned nodes. The generalization to the case of  more than two pinned nodes is straightforward and is omitted here for conciseness.
{We recall that the matrix $T_\Pi$ has a special structure, where each one of its rows has nonzero entries only corresponding to the pinned nodes that are swapped by the PNS.}
We can then associate a `quotient network' to the PNS where all the pinned nodes swapped by the PNS are mapped to only one node. By using Theorem \ref{lemmaquotient} and Remark \ref{R}, we see that the transformation $T_\Pi$  decouples the pair  $(L,R)$ into a quotient block with one pinned node and an orthogonal (redundant) block with the other pinned node. 
\end{proof}

As an example, consider the $N=10$ node weighted network represented in Fig.\ \ref{fig:Abu} (a), with nodes color coded according to the $M=8$ equitable clusters to which they belong. Thin (thick) edges indicate a weight of the connection equal to one (two.) The three pinned nodes ($s=3$), number $1$, $2$, and $6$ are indicated with a red arrow pointing at them. 
This is an example of a network for which the quotient network is not fully controllable.   Application of the transformation matrix $\hat{T}$ to the pair $(L, R)$  decouples the pair $(L, R)$ into four block pairs: a $7$-dimensional qc pair (shown in black), a $1$-dimensional rc pair (in red), a $1$-dimensional qu pair (in green) and a $1$-dimensional ru pair (in blue). 
\begin{figure}[H]
    \centering
    \includegraphics[scale=.36]{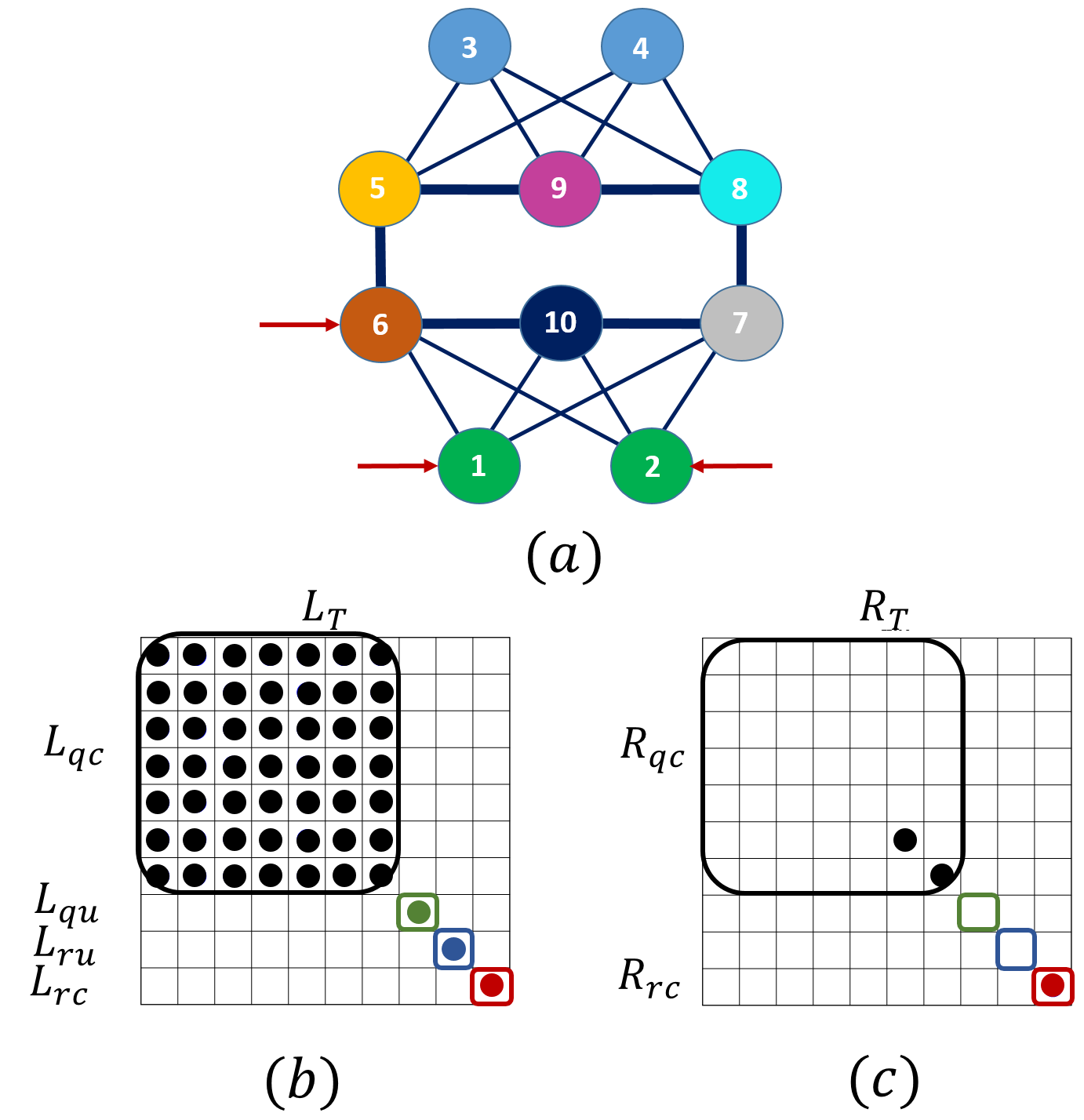}
    \caption{(a) $N=10$-node weighted network. Nodes are colored according to the $M=8$ equitable clusters to which they belong. Thin (thick) edges indicate a weight of the connection equal to one (two.) Application of transformation matrix $\hat{T}$ to (b) $L$ and (c) $R$ decouples the pair $(L, R)$ into the qc block pair shown in black, the rc block pair shown in red, the qu block pair shown in green, and the ru block pair shown in dark blue.}
    \label{fig:Abu}
\end{figure}

It can be seen that this network has one PNS, which corresponds to a disjoint cycle that swaps node 1 with 2, $l = 2$,
\begin{equation}
    \Pi=\begin{pmatrix}
    0 & 1 & 0 & 0 & 0 & 0 & 0\\
    1 & 0 & 0 & 0 & 0 & 0 & 0\\
    0 & 0 & 1 & 0 & 0 & 0 & 0\\
    0 & 0 & 0 & 1 & 0 & 0 & 0\\
    0 & 0 & 0 & 0 & 1 & 0 & 0\\
    0 & 0 & 0 & 0 & 0 & 1 & 0\\
    0 & 0 & 0 & 0 & 0 & 0 & 1
    \end{pmatrix}
\end{equation}
The block-diagonal transformation matrix $T_\Pi$  {whose columns are} the eigenvectors of matrix $\Pi$ is,
\begin{equation}
    T_\Pi=\begin{pmatrix}
-1/\sqrt{2} & 1/\sqrt{2} & 0 & 0 & 0 & 0 & 0 & 0 & 0 & 0 \\ 
1/\sqrt{2} & 1/\sqrt{2} & 0 & 0 & 0 & 0 & 0 & 0 & 0 & 0 \\ 
0 & 0 & 1 & 0 & 0 & 0 & 0 & 0 & 0 & 0 \\ 
0 & 0 & 0 & 1 & 0 & 0 & 0 & 0 & 0 & 0 \\ 
0 & 0 & 0 & 0 & 1 & 0 & 0 & 0 & 0 & 0 \\ 
0 & 0 & 0 & 0 & 0 & 1 & 0 & 0 & 0 & 0 \\ 
0 & 0 & 0 & 0 & 0 & 0 & 1 & 0 & 0 & 0 \\ 
0 & 0 & 0 & 0 & 0 & 0 & 0 & 1 & 0 & 0 \\ 
0 & 0 & 0 & 0 & 0 & 0 & 0 & 0 & 1 & 0 \\ 
0 & 0 & 0 & 0 & 0 & 0 & 0 & 0 & 0 & 1 
    \end{pmatrix}
\end{equation}
Now by applying the transformation matrix $T_\Pi$ to the pair $(L, R)$, we have:
\begin{equation}
\begin{array}{clc}
L_\Pi=\left( \resizebox{.85\hsize}{!}{$
\begin{array}{c | c c c c c c c c c} 
-3 & 0 & 0 & 0 & 0 & 0 & 0 & 0 & 0 & 0 \\ 
\hline
0 & -3 & 0 & 0 & 0 & \sqrt{2} & \sqrt{2} & 0 & 0 & \sqrt{2} \\ 
0 & 0 & -3 & 0 & 1 & 0 & 0 & 1 & 1 & 0 \\ 
0 & 0 & 0 & -3 & 1 & 0 & 0 & 1 & 1 & 0 \\ 
0 & 0 & 1 & 1 & -6 & 2 & 0 & 0 & 2 & 0 \\ 
0 & \sqrt{2} & 0 & 0 & 2 & -6 & 0 & 0 & 0 & 2 \\ 
0 & \sqrt{2} & 0 & 0 & 0 & 0 & -6 & 2 & 0 & 2 \\ 
0 & 0 & 1 & 1 & 0 & 0 & 2 & -6 & 2 & 0 \\ 
0 & 0 & 1 & 1 & 2 & 0 & 0 & 2 & -6 & 0 \\ 
0 & \sqrt{2} & 0 & 0 & 0 & 2 & 2 & 0 & 0 & -6 
\end{array}$}
\right)\\\\
    R_\Pi=\left( \resizebox{.45\hsize}{!}{$
    \begin{array}{c | c c c c c c c c c}
1 & 0 & 0 & 0 & 0 & 0 & 0 & 0 & 0 & 0 \\ 
\hline
0 & 1 & 0 & 0 & 0 & 0 & 0 & 0 & 0 & 0 \\ 
0 & 0 & 0 & 0 & 0 & 0 & 0 & 0 & 0 & 0 \\ 
0 & 0 & 0 & 0 & 0 & 0 & 0 & 0 & 0 & 0 \\ 
0 & 0 & 0 & 0 & 0 & 0 & 0 & 0 & 0 & 0 \\ 
0 & 0 & 0 & 0 & 0 & 1 & 0 & 0 & 0 & 0 \\ 
0 & 0 & 0 & 0 & 0 & 0 & 0 & 0 & 0 & 0 \\ 
0 & 0 & 0 & 0 & 0 & 0 & 0 & 0 & 0 & 0 \\ 
0 & 0 & 0 & 0 & 0 & 0 & 0 & 0 & 0 & 0 \\ 
0 & 0 & 0 & 0 & 0 & 0 & 0 & 0 & 0 & 0 
    \end{array}$}
    \right)
    \end{array}
    \end{equation}

We see that the application of $T_\Pi$ {to} the network shown in Fig.\ \ref{fig:Abu}, decouples the pair $(L, R)$ into a $9$-dimensional driven quotient block and a $1$-dimensional driven redundant block. 

{
To conclude, Lemma \ref{lemma7} states that choosing two pinned nodes which are swapped by a PNS leads to two independent driven blocks and confirms that both the topology of the network and the number and choice of the pinned nodes affect the number and dimensions of the driven pairs.}

\section{Application {to Complex} networks}



{In the previous sections we have applied our approach to a variety of simple networks with a small number of nodes. Here, we extend our analysis to larger synthetic networks and complex networks with topologies taken from the literature. The main goal of this section is to provide evidence of what appears to be a generic property of these networks in terms of the structure of the blocks that arise from the decomposition.} Table \ref{table:realnet} provides for each network information on the number of nodes $N$, the number of edges $E$, the number of pinned nodes $s$, and $u$ the dimension of the largest block which corresponds to the only quotient controllable block. All these networks are undirected and unweighted. {For each network in the table, $s=7$ (number chosen without loss of generality) pinned nodes are randomly selected.}

{The first and second networks ($SF_1$ and $SF_2$) are synthetic scale-free networks generated with the static model \cite{korea},  with exponents of the power-law degree distribution equal to $\alpha=2.11$ and $\alpha=2.56$, respectively. In both cases, we first generated the network connectivity using the static model, then removed all nodes that were not part of the network giant component.  The third dataset, DD-g147, is a biological protein interaction network \cite{nr}. The fourth dataset, ca-sandi-auths, represents a collaboration network between scientists at Sandia National Labs \cite{nr}. The last dataset, Case60nordic is a power grid network \cite{5491276}.} 


To demonstrate the use of the method to determine the stability of the synchronous solution, we consider the R\"{o}ssler oscillator (Eq. \ref{eq:ross}) at each network node and set $\bG(\bx_i)=[x_i,y_i,0]^T$ and $\bH(\bx_i)=[0,y_i,0]^T$.  For this choice of the functions $\bF, \bG, \bH$, the target synchronous solution is stable for $\gamma$ larger than a critical value $\gamma^*_1$. 
We emphasize that the choice of the specific oscillator is not relevant for the goal specified above, and we have obtained similar results for other choices of the functions $\bF, \bG, \bH$.  Table \ref{table:realnet} compares 
the {numerically estimated} values of
$\gamma^*_1$ and $\gamma^*_2$, where the latter is the critical value of $\gamma$ above which the MLE of the largest block of the pair ($L_T, R_T$) becomes negative. As can be seen, for all these examples the MLE corresponding to the largest block determines {the} stability of the target solution. 

{Figure \ref{fig:4net} compares -- by plotting them against the pinning control coupling coefficient $\gamma$ -- the synchronization error of the whole network (panels a,c,e,g,i) and the MLE of the largest block of the pair ($L_T, R_T$), whereas the MLEs of the remaining blocks are all negative (panels b,d,f,h,j).}

Figure \ref{error} shows the values of $\gamma^*_1$ and $\gamma^*_2$ versus the number of pinned nodes $s$ for the last network listed in Table \ref{table:realnet} ($60nordic$). For each value of $s$, the pinned nodes are randomly chosen. As can be seen from Fig.\ \ref{error},   $\gamma_1^*$ and $\gamma_2^*$ are found to be in close agreement for all values of $s$.    Figure \ref{error} provides evidence that, independent of the number of pinned nodes, the stability of the target synchronous solution is solely a function of the MLE of the largest block of the pair $(L_T, R_T)$. This has two main implications: on the one hand this is advantageous because stability of the target synchronous solution is only determined by this one block, on the other hand it is disadvantageous because this one block is typically quite large, indicating that our ability to reduce the dimension of the problem is intrinsically limited for these networks. 

{We would like to emphasize that the
purpose of this section was not to show application of the pinning control problem to examples of practical interest, but rather to show general properties of the blocks obtained when randomly selecting the pin nodes of a {large} complex network. In this respect, we
found the remarkable property that for all the networks considered ({synthetic and `real'}) {and for all the random selections of the pin nodes,} there
is typically one larger quotient and controllable block, and the ability to drive the
network to the target solution is \textit{de facto} determined by the stability of this one block. }

\begin{table}[H]
\centering
\caption{
For each network  we report $N$, ${E}$, $s$, and $u$, the total number of nodes, the number of edges, the number of pinned nodes, and the dimension of the $qc$ block, respectively. $\gamma^*_1$ is the value of $\gamma$ above which the network converges to the synchronous solution, and $\gamma^*_2$ is the critical value of  $\gamma$ above which the MLE of the largest block of the pair ($L_T, R_T$) becomes negative. All these networks are undirected and unweighted.}\label{table:realnet}
\begin{tabular}{|c||c|c|c|c|c|c|c|}
\hline
     Label & Name & $N$ & ${E}$ & $s$ & $u$ & $\gamma^*_1$ & $\gamma^*_2$\\
     \hline
     \hline
     $Net_1$ & $SF_1$ & 72 & 100 & 7 & 46 & 1.43 & 1.425\\
     \hline
     $Net_2$ & $SF_2$ & 63 & 80 & 7 & 55 & 1.43 & 1.415\\
     \hline
     $Net_3$ & DD-g147\cite{nr} & 99 & 317 & 7 & 94 & 2.385 & 2.368\\
     \hline
     $Net_4$ & ca-sandi-auths\cite{nr} & 56 & 70 & 7 & 13 & 1.27 & 1.268\\
     \hline
     $Net_5$ & Case60nordic\cite{capitanescu}& 60 & 72 & 7 & 56 & 1.302 & 1.291\\
     \hline
     \end{tabular}
\end{table}
\begin{figure}[H]
    \centering
     \subfigure{\includegraphics[width=0.9\linewidth]{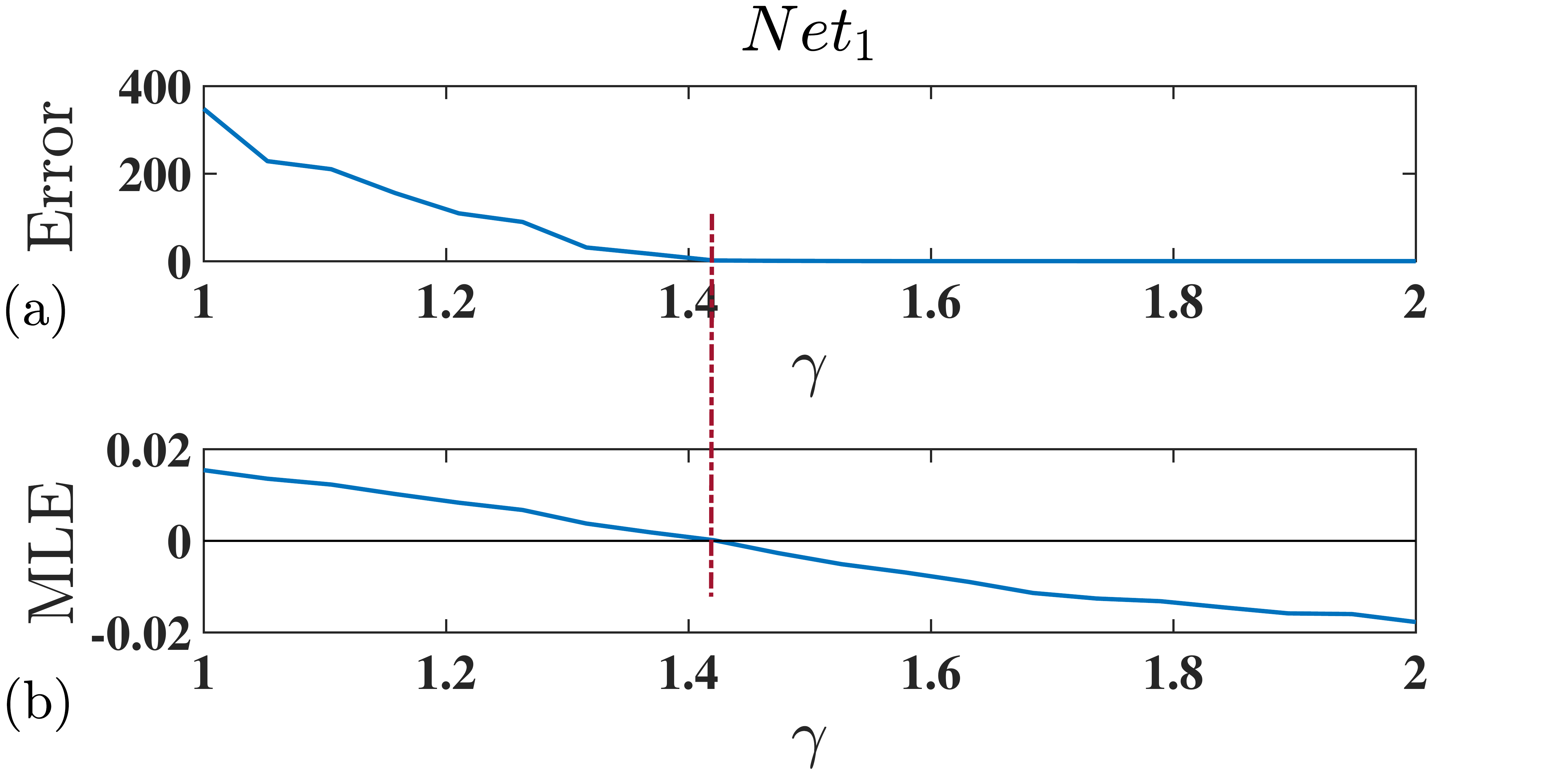}}\\
     \subfigure{\includegraphics[width=0.9\linewidth]{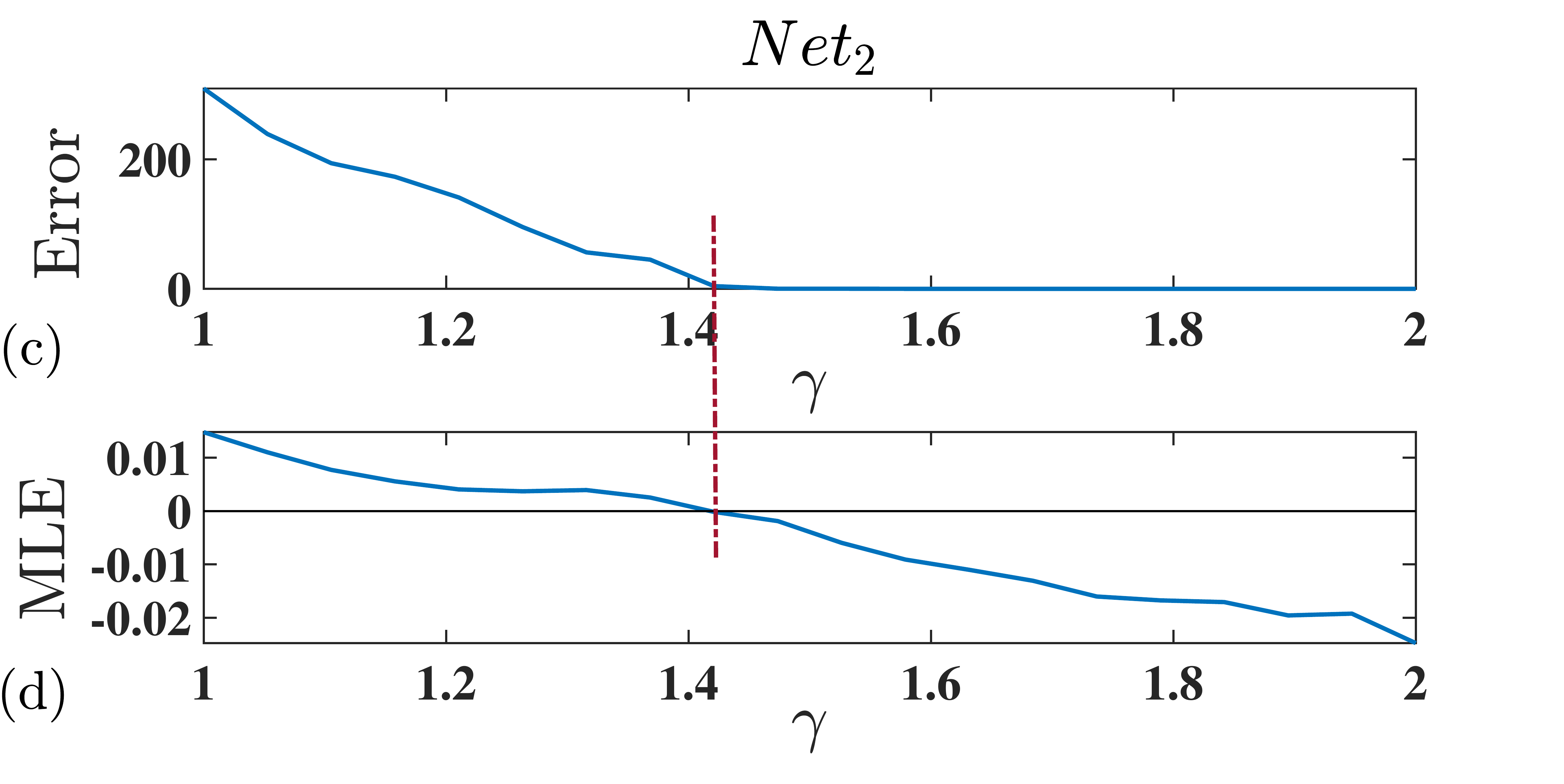}}\\
       \subfigure{\includegraphics[width=0.9\linewidth]{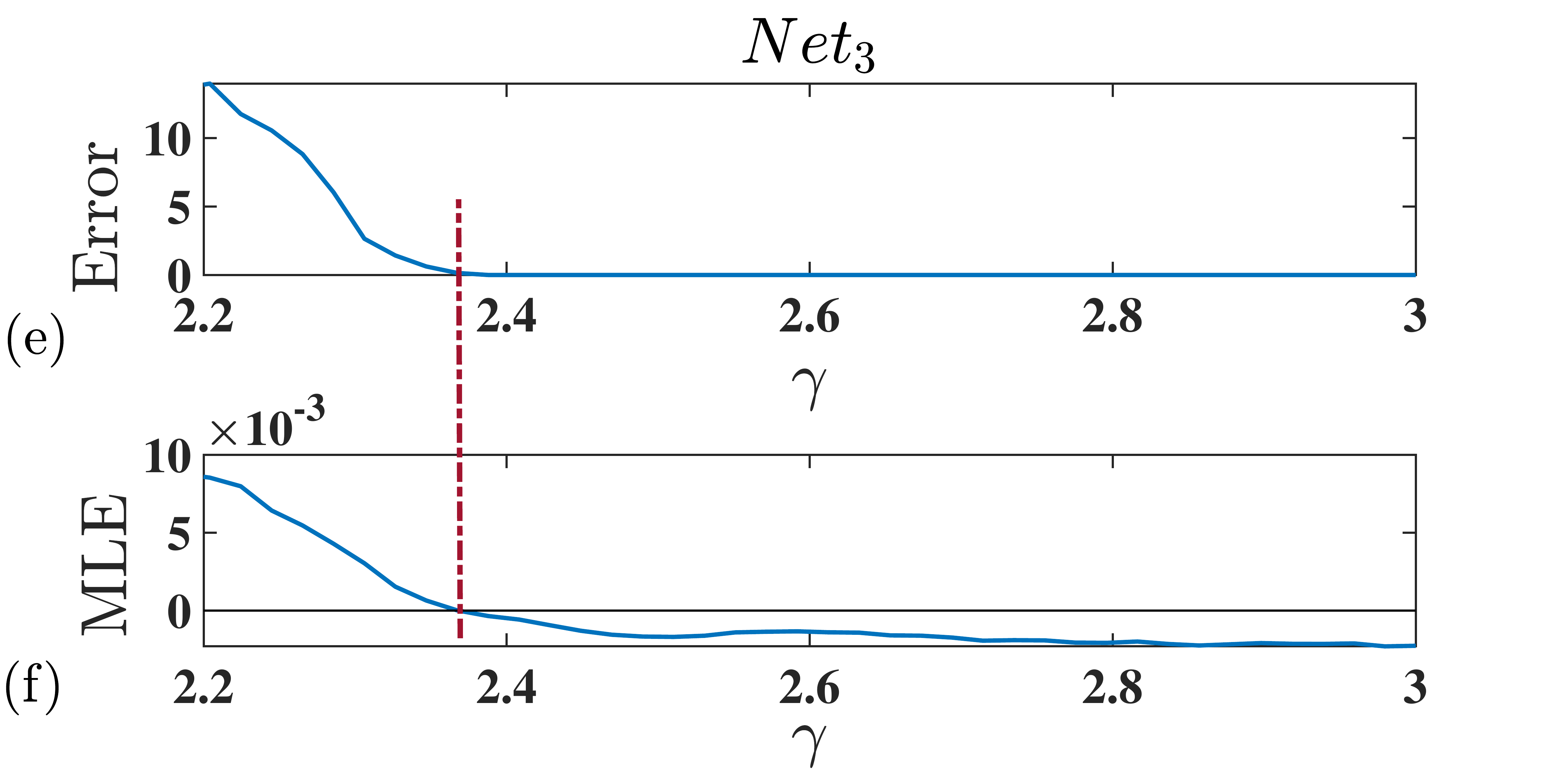}}\\
        \subfigure{\includegraphics[width=0.9\linewidth]{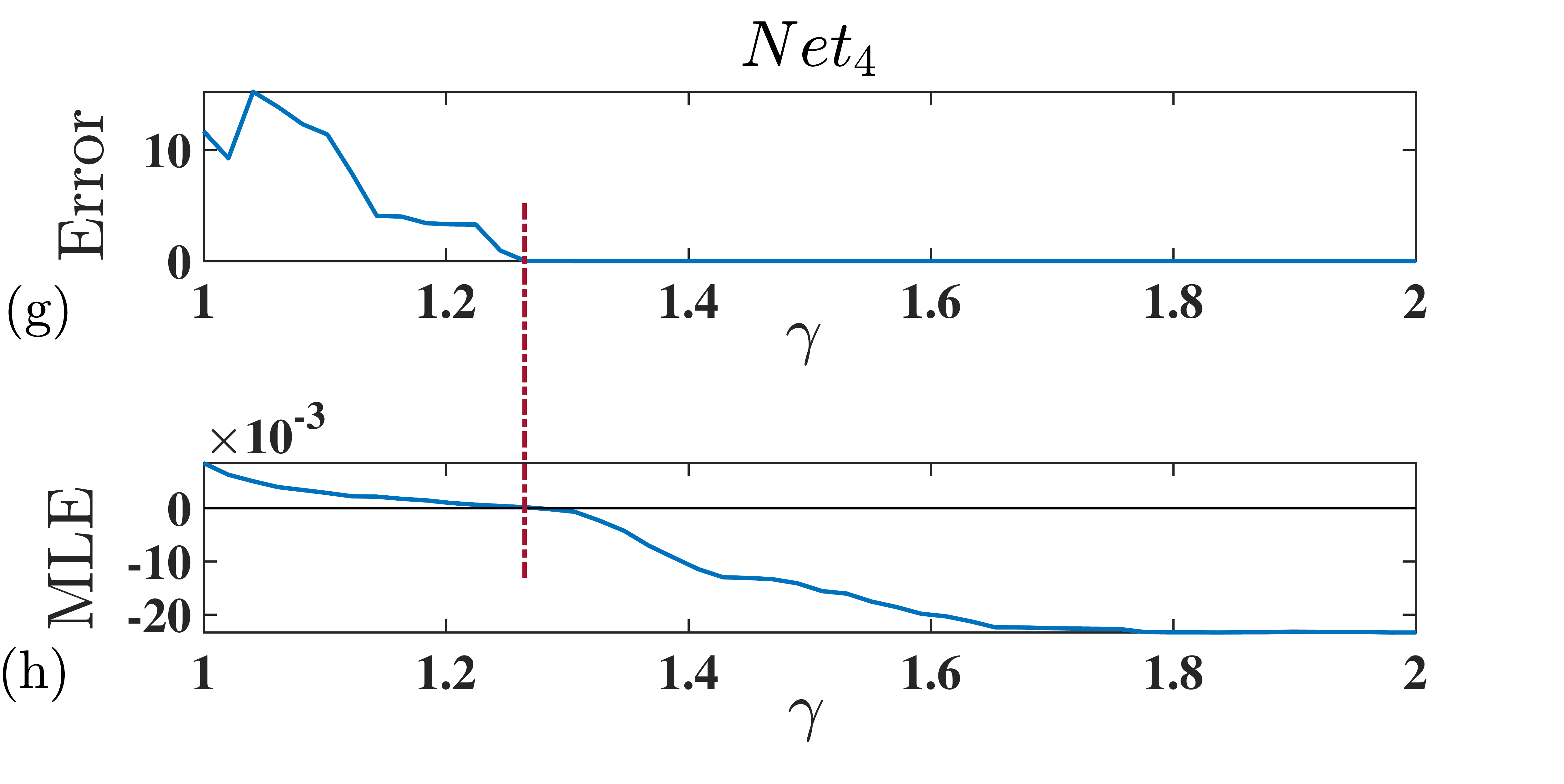}}\\
         \subfigure{\includegraphics[width=0.9\linewidth]{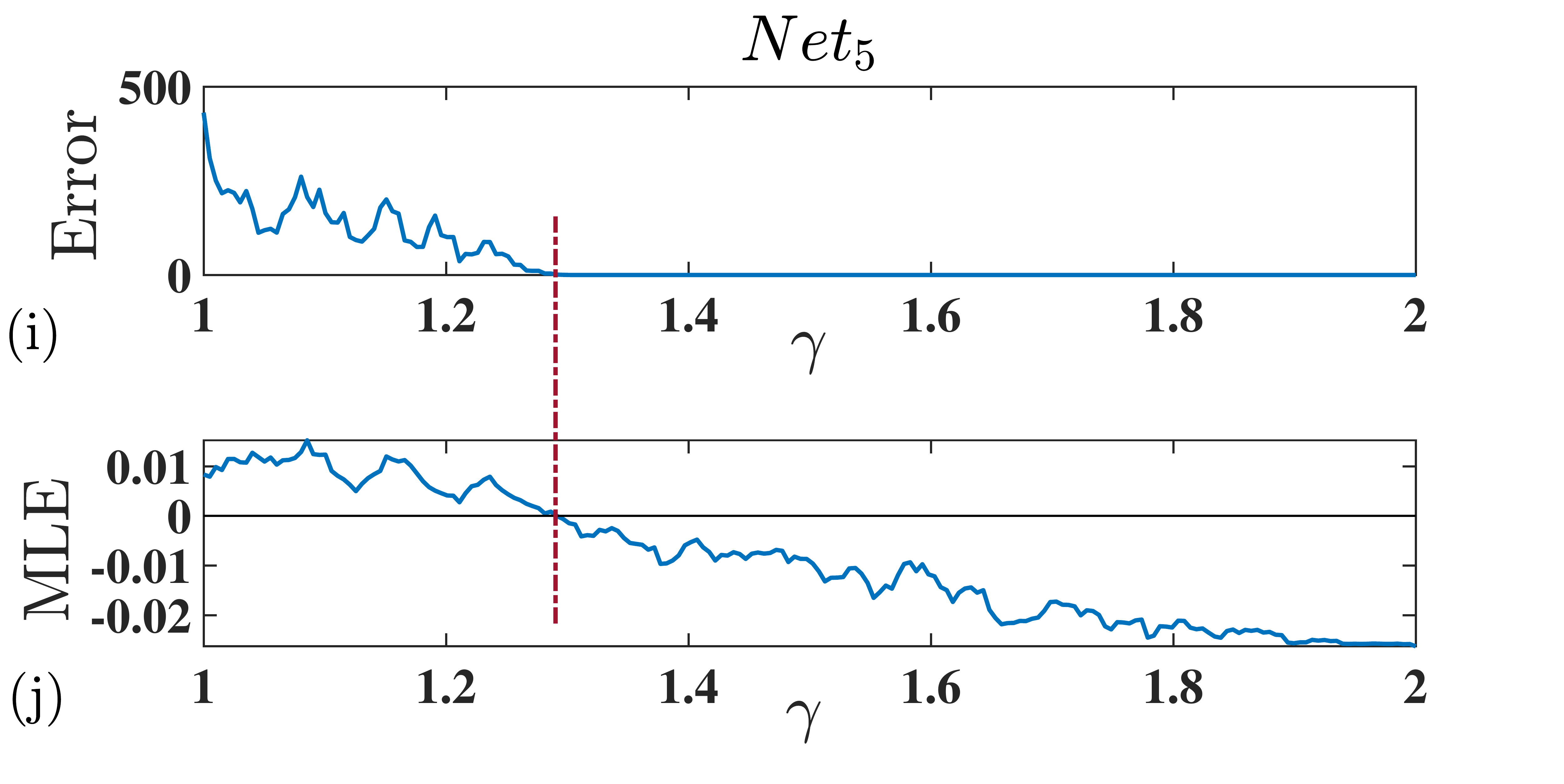}}
    \caption{(a,c,e,g,i) Synchronization error for the pair ($L, R$) as a function of $\gamma$. (b,d,f,h,j) MLE of the largest block of the pair ($L_T, R_T$) as a function of $\gamma$. The network labels $Net_1$, $Net_2$, etc are consistent with Table II.}
    \label{fig:4net}
\end{figure}
\begin{figure}[H]
    \centering
    \includegraphics[width=\linewidth]{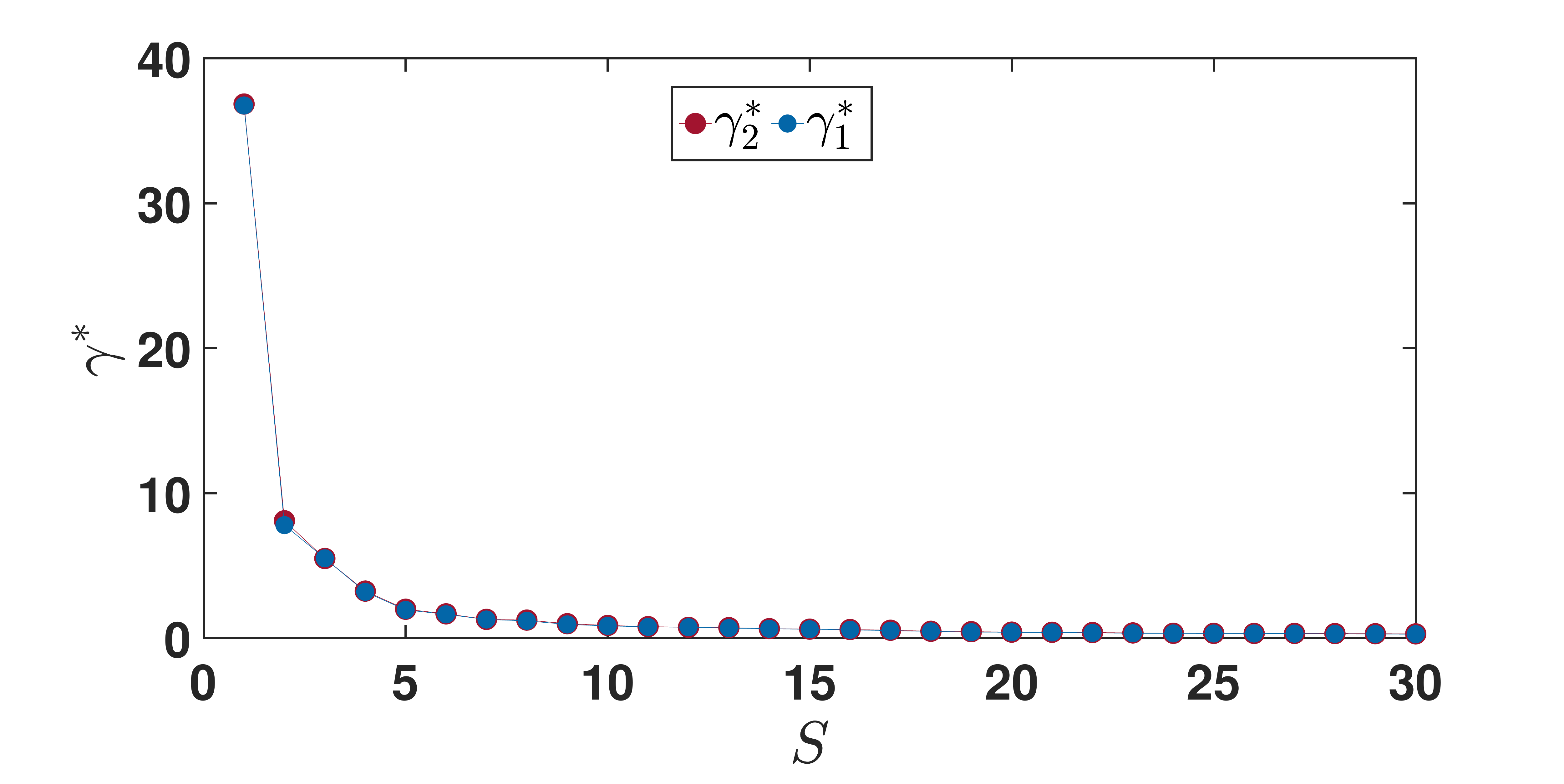}
    \caption{$\gamma^*_1$ and $\gamma^*_2$ versus the number of pinned nodes for the $60nordic$ network listed in Table \ref{table:realnet}, where $\gamma^*_1$ is the value of the $\gamma$ above which the error goes to zero, and $\gamma^*_2$ is the zero-crossing point of the MLE of the largest block of the pair ($L_T, R_T$).}
    \label{error}
\end{figure}

\color{black}

\section{Conclusion}\label{conclusion}
While techniques for the simultaneous block diagonalization (SBD) of matrices have been previously applied to the network synchronization problem \cite{Ir:So,zhang2020},
{here, our focus is on decoupling the synchronization stability equations of a network with 
pinning control, 
consisting of two different types of coupling: (i) node-to-node coupling among the network nodes and (ii) input-to-node coupling from source node to the `pinned nodes'.}
 Our main result is that {we prove that the blocks resulting from the SBD can be categorized into four types}, which we call quotient controllable, quotient uncontrollable, redundant controllable, redundant uncontrollable. This has important consequences as it indicates that, for this class of networks, the SBD transformation can be replaced by application of an alternative transformation that decouples the stability problem into a quotient and a redundant part and each one of these two into a controllable and an uncontrollable part. {Different from previous applications of the SBD technique \cite{Ir:So,zhang2020}, we provide a characterization about the dimensions of the independent sets of equations in which the stability problem is reduced.} 
 Our analysis applied to several complex networks from the literature shows that stability of the target synchronous solution is always determined by the maximum Lyapunov exponent of only one quotient controllable block. 

{From the standpoint of stability analysis, it is convenient to have small blocks.} One may conclude that it may be desirable to minimize the dimension of the controllable subspace, as we have shown that the blocks corresponding to the non-controllable subspace are all scalar, which is the best possible outcome in terms of dimension reduction.
In terms of stability, we have investigated how the synchronizability (or pin-controllability) varies depending on the individual blocks. One point that is left for future study is how to choose the pinned nodes in an optimal way, i.e., so as to enforce desirable properties of the controlled network. {It is also possible to extend our work to the case that delays affect either the node-to-node coupling function $\bG$ or the source-to-node coupling function $\bH$. The techniques for simultaneous block diagonalization of matrices can also be applied to the case of time-varying networks \cite{dariani2011effect,feng2016cluster,lin2021pinning,shi2021synchronization}, but with the caveat that the blocks resulting from the decomposition may vary in time with the network time evolution. }

\section*{Supplementary Material}
{
The supplementary material includes further clarifications on the calculation and then application of the transformation matrices $T$ and $\hat{T}$. Note 1 provides further information on the network with $N=10$ nodes shown in Figure 5. Note 2 focuses on another example of a star network with $N=4$ nodes, for which we considered all possible choices of unique combinations of pinned nodes.
Note 3 includes the proof that that if $T^{(1)}, ... ,T^{(M)}$ is a set of $N \times N$ real matrices, then the intersection of their null spaces is equal to the null space of the sum of their inner products.
}
\section*{Acknowledgement} The authors are in debt to Isaac Klickstein and Galen Novello for insightful discussions.
\section*{Data Availability}
The data that support the findings of this study are available
within the article.

\section{Appendix}
\setcounter{equation}{0}
\renewcommand\theequation{A.\arabic{equation}}
\subsection{\texorpdfstring{Derivation of Blocks P\textsubscript{1} and P\textsubscript{2}}{}}\label{Ap1}

 Matrix-matrix products and be expressed as matrix-vector products using the Kronecker product.
Let $A \in \mathbb{R}^{n \times p}$, $X \in \mathbb{R}^{p \times m}$ and $B \in \mathbb{R}^{n \times m}$.
The matrix-matrix product $AX = B$ can be equivalently expressed as,

\begin{equation}
    (I_m \otimes A) \text{vec}(X) = \text{vec}(B)
\end{equation}
Similarly, for $A \in \mathbb{R}^{p \times m}$ and $X \in \mathbb{R}^{n \times p}$ then the matrix-matrix product $XA = B$ can be equivalently expressed as,
\begin{equation}
    (A^T \otimes I_n) \text{vec}(X) = \text{vec}(B)
\end{equation}
With these two identities, the four equations in Eq. \eqref{12} can be equivalently expressed as matrix-vector product.
\begin{subequations}\label{eq:LP_kron}
    \begin{equation}
        ((I_s \otimes L_{11}) - (L_{11}^T \otimes I_s)) \text{vec}(P_1) = \boldsymbol{0}_{s^2}
    \end{equation}
    \begin{equation}
        ((I_\tau \otimes L_{22}) - (L_{22}^T \otimes I_\tau)) \text{vec}(P_2) = \boldsymbol{0}_{\tau^2}
    \end{equation}
    \begin{equation}
        (L_{12}^T \otimes I_s ) \text{vec}(P_1) - (I_\tau \otimes L_{12}) \text{vec}(P_2) = \boldsymbol{0}_{s\tau}
    \end{equation}
    \begin{equation}
        -(I_s \otimes L_{21}) \text{vec}(P_1) + (L_{21}^T \otimes I_\tau) \text{vec}(P_2) = \boldsymbol{0}_{s\tau}
    \end{equation}
\end{subequations}

Combine the first two lines of Eq. \eqref{eq:LP_kron} into a composite linear system of equations and the second two lines of Eq. \eqref{eq:LP_kron} into another composite linear system of equations to create Eq. \eqref{eq:nullspaces}.

%
\bibliography{main}

\begin{thebibliography}{10}
\providecommand{\url}[1]{#1}
\csname url@samestyle\endcsname
\providecommand{\newblock}{\relax}
\providecommand{\bibinfo}[2]{#2}
\providecommand{\BIBentrySTDinterwordspacing}{\spaceskip=0pt\relax}
\providecommand{\BIBentryALTinterwordstretchfactor}{4}
\providecommand{\BIBentryALTinterwordspacing}{\spaceskip=\fontdimen2\font plus
\BIBentryALTinterwordstretchfactor\fontdimen3\font minus
  \fontdimen4\font\relax}
\providecommand{\BIBforeignlanguage}[2]{{%
\expandafter\ifx\csname l@#1\endcsname\relax
\typeout{** WARNING: IEEEtran.bst: No hyphenation pattern has been}%
\typeout{** loaded for the language `#1'. Using the pattern for}%
\typeout{** the default language instead.}%
\else
\language=\csname l@#1\endcsname
\fi
#2}}
\providecommand{\BIBdecl}{\relax}
\BIBdecl

\bibitem{wu2008cluster}
W.~Wu, W.~Zhou, and T.~Chen, ``Cluster synchronization of linearly coupled
  complex networks under pinning control,'' \emph{IEEE Transactions on Circuits
  and Systems I: Regular Papers}, vol.~56, no.~4, pp. 829--839, 2008.

\bibitem{liu2011cluster}
X.~Liu and T.~Chen, ``Cluster synchronization in directed networks via
  intermittent pinning control,'' \emph{IEEE Transactions on Neural Networks},
  vol.~22, no.~7, pp. 1009--1020, 2011.

\bibitem{su2012decentralized}
H.~Su, Z.~Rong, M.~Z. Chen, X.~Wang, G.~Chen, and H.~Wang, ``Decentralized
  adaptive pinning control for cluster synchronization of complex dynamical
  networks,'' \emph{IEEE Transactions on Cybernetics}, vol.~43, no.~1, pp.
  394--399, 2012.

\bibitem{yu2013synchronization}
W.~Yu, G.~Chen, J.~Lu, and J.~Kurths, ``Synchronization via pinning control on
  general complex networks,'' \emph{SIAM Journal on Control and Optimization},
  vol.~51, no.~2, pp. 1395--1416, 2013.

\bibitem{liu2014synchronization}
X.~Liu and T.~Chen, ``Synchronization of nonlinear coupled networks via
  aperiodically intermittent pinning control,'' \emph{IEEE Transactions on
  Neural Networks and Learning Systems}, vol.~26, no.~1, pp. 113--126, 2014.

\bibitem{wang2015pinning}
J.-L. Wang, H.-N. Wu, T.~Huang, S.-Y. Ren, and J.~Wu, ``Pinning control for
  synchronization of coupled reaction-diffusion neural networks with directed
  topologies,'' \emph{IEEE Transactions on Systems, Man, and Cybernetics:
  Systems}, vol.~46, no.~8, pp. 1109--1120, 2015.

\bibitem{liu2015synchronization}
X.~Liu and T.~Chen, ``Synchronization of complex networks via aperiodically
  intermittent pinning control,'' \emph{IEEE transactions on automatic
  control}, vol.~60, no.~12, pp. 3316--3321, 2015.

\bibitem{delellis2018partial}
P.~DeLellis, F.~Garofalo, and F.~L. Iudice, ``The partial pinning control
  strategy for large complex networks,'' \emph{Automatica}, vol.~89, pp.
  111--116, 2018.

\bibitem{belykh2008cluster}
V.~N. Belykh, G.~V. Osipov, V.~S. Petrov, J.~A. Suykens, and J.~Vandewalle,
  ``Cluster synchronization in oscillatory networks,'' \emph{Chaos: An
  Interdisciplinary Journal of Nonlinear Science}, vol.~18, no.~3, p. 037106,
  2008.

\bibitem{belykh2000hierarchy}
V.~N. Belykh, I.~V. Belykh, and M.~Hasler, ``Hierarchy and stability of
  partially synchronous oscillations of diffusively coupled dynamical
  systems,'' \emph{Physical Review E}, vol.~62, no.~5, p. 6332, 2000.

\bibitem{chen2022pinning}
G.~Chen, ``Pinning control of complex dynamical networks,'' \emph{IEEE
  Transactions on Consumer Electronics}, 2022.

\bibitem{tang2013}
Y.~Tang, H.~Gao, and J.~Kurths, ``Distributed robust synchronization of
  dynamical networks with stochastic coupling,'' \emph{IEEE Transactions on
  Circuits and Systems I: Regular Papers}, vol.~61, no.~5, pp. 1508--1519,
  2013.

\bibitem{chen2014pinning}
G.~Chen, ``Pinning control and synchronization on complex dynamical networks,''
  \emph{International Journal of Control, Automation and Systems}, vol.~12,
  no.~2, pp. 221--230, 2014.

\bibitem{liu2015pinning}
W.~Liu, W.~Gu, W.~Sheng, X.~Meng, S.~Xue, and M.~Chen, ``Pinning-based
  distributed cooperative control for autonomous microgrids under uncertain
  communication topologies,'' \emph{IEEE Transactions on Power Systems},
  vol.~31, no.~2, pp. 1320--1329, 2015.

\bibitem{orouskhani2016optimizing}
Y.~Orouskhani, M.~Jalili, and X.~Yu, ``Optimizing dynamical network structure
  for pinning control,'' \emph{Scientific reports}, vol.~6, 2016.

\bibitem{chen2018pinning}
H.~Chen, P.~Shi, and C.-C. Lim, ``Pinning impulsive synchronization for
  stochastic reaction--diffusion dynamical networks with delay,'' \emph{Neural
  Networks}, vol. 106, pp. 281--293, 2018.

\bibitem{dariani2011effect}
R.~Dariani, A.~Buscarino, L.~Fortuna, and M.~Frasca, ``Effect of motion on
  pinning control of time-varying networks,'' in \emph{Neural Nets
  WIRN11}.\hskip 1em plus 0.5em minus 0.4em\relax IOS Press, 2011, pp.
  105--111.

\bibitem{feng2016cluster}
J.~Feng, P.~Yang, and Y.~Zhao, ``Cluster synchronization for nonlinearly
  time-varying delayed coupling complex networks with stochastic perturbation
  via periodically intermittent pinning control,'' \emph{Applied Mathematics
  and Computation}, vol. 291, pp. 52--68, 2016.

\bibitem{lin2021pinning}
H.~Lin and J.~Wang, ``Pinning control of complex networks with time-varying
  inner and outer coupling,'' \emph{Mathematical Biosciences and Engineering},
  vol.~18, no.~4, pp. 3435--3447, 2021.

\bibitem{shi2021synchronization}
L.~Shi, C.~Zhang, and S.~Zhong, ``Synchronization of singular complex networks
  with time-varying delay via pinning control and linear feedback control,''
  \emph{Chaos, Solitons \& Fractals}, vol. 145, p. 110805, 2021.

\bibitem{an:2011}
Z.~An, H.~Zhu, X.~Li, C.~Xu, Y.~Xu, and X.~Li, ``Nonidentical linear
  pulse-coupled oscillators model with application to time synchronization in
  wireless sensor networks,'' \emph{IEEE Transactions on Industrial
  Electronics}, vol.~58, no.~6, pp. 2205--2215, 2011.

\bibitem{kan:2015}
Z.~Kan, J.~R. Klotz, E.~L. Pasiliao~Jr, and W.~E. Dixon, ``Containment control
  for a social network with state-dependent connectivity,'' \emph{Automatica},
  vol.~56, pp. 86--92, 2015.

\bibitem{trentelman:2013}
H.~L. Trentelman, K.~Takaba, and N.~Monshizadeh, ``Robust synchronization of
  uncertain linear multi-agent systems,'' \emph{IEEE Transactions on Automatic
  Control}, vol.~58, no.~6, pp. 1511--1523, 2013.

\bibitem{russo:2009}
G.~Russo and M.~D. Bernardo, ``How to synchronize biological clocks,''
  \emph{Journal of Computational Biology}, vol.~16, no.~2, pp. 379--393, 2009.

\bibitem{su:2013}
H.~Su and X.~Wang, \emph{Pinning control of complex networked systems:
  Synchronization, consensus and flocking of networked systems via
  pinning}.\hskip 1em plus 0.5em minus 0.4em\relax Springer Science \& Business
  Media, 2013.

\bibitem{wu2009}
Y.~Wu, W.~Wei, G.~Li, and J.~Xiang, ``Pinning control of uncertain complex
  networks to a homogeneous orbit,'' \emph{IEEE Transactions on Circuits and
  Systems II: Express Briefs}, vol.~56, no.~3, pp. 235--239, 2009.

\bibitem{barajas2018}
J.~G. Barajas-Ram{\'\i}rez, ``Adaptive tracking via pinning in networks of
  nonidentical nodes,'' \emph{Kybernetika}, vol.~54, no.~1, pp. 30--40, 2018.

\bibitem{vega2018}
C.~J. Vega, E.~N. Sanchez, and R.~Alzate, ``Inverse optimal pinning control for
  synchronization of complex networks with nonidentical chaotic nodes,''
  \emph{IFAC-PapersOnLine}, vol.~51, no.~13, pp. 235--239, 2018.

\bibitem{alanis2021}
A.~Y. Alanis, D.~R{\'\i}os-Rivera, E.~N. Sanchez, and O.~D. Sanchez, ``Learning
  impulsive pinning control of complex networks,'' \emph{Mathematics}, vol.~9,
  no.~19, p. 2436, 2021.

\bibitem{uhlig:1973}
F.~Uhlig, ``Simultaneous block diagonalization of two real symmetric
  matrices,'' \emph{Linear Algebra and Its Applications}, vol.~7, no.~4, pp.
  281--289, 1973.

\bibitem{maehara2010numerical}
T.~Maehara and K.~Murota, ``A numerical algorithm for block-diagonal
  decomposition of matrix $*$-algebras with general irreducible components,''
  \emph{Japan journal of industrial and applied mathematics}, vol.~27, no.~2,
  pp. 263--293, 2010.

\bibitem{murota2010numerical}
K.~Murota, Y.~Kanno, M.~Kojima, and S.~Kojima, ``A numerical algorithm for
  block-diagonal decomposition of matrix $*$-algebras with application to
  semidefinite programming,'' \emph{Japan Journal of Industrial and Applied
  Mathematics}, vol.~27, no.~1, pp. 125--160, 2010.

\bibitem{Ir:So}
D.~Irving and F.~Sorrentino, ``Synchronization of a hypernetwork of coupled
  dynamical systems,'' \emph{Phys. Rev. E}, vol.~86, p. 056102, 2012.

\bibitem{zhang2020}
Y.~Zhang and A.~E. Motter, ``Symmetry-independent stability analysis of
  synchronization patterns,'' \emph{SIAM Review}, vol.~62, no.~4, pp. 817--836,
  2020.

\bibitem{sorrentino2007}
F.~Sorrentino, M.~Di~Bernardo, F.~Garofalo, and G.~Chen, ``Controllability of
  complex networks via pinning,'' \emph{Physical Review E}, vol.~75, no.~4, p.
  046103, 2007.

\bibitem{sorrentino2007effects}
F.~Sorrentino, ``Effects of the network structural properties on its
  controllability,'' \emph{Chaos: An Interdisciplinary Journal of Nonlinear
  Science}, vol.~17, no.~3, p. 033101, 2007.

\bibitem{Pe:Ca}
L.~Pecora and T.~Carroll, ``Master stability functions for synchronized coupled
  systems,'' \emph{Phys. Rev. Lett.}, vol.~80, pp. 2109--2112, 1998.

\bibitem{HYP}
F.~Sorrentino, ``Synchronization of hypernetworks of coupled dynamical
  systems,'' \emph{New J. Phys.}, vol.~14, p. 033035, 2012.

\bibitem{maehara2011algorithm}
T.~Maehara and K.~Murota, ``Algorithm for error-controlled simultaneous
  block-diagonalization of matrices,'' \emph{SIAM Journal on Matrix Analysis
  and Applications}, vol.~32, no.~2, pp. 605--620, 2011.

\bibitem{ogata2010modern}
K.~Ogata, \emph{Modern control engineering}.\hskip 1em plus 0.5em minus
  0.4em\relax Prentice hall, 2010.

\bibitem{kailath1980linear}
T.~Kailath, \emph{Linear systems}.\hskip 1em plus 0.5em minus 0.4em\relax
  Prentice-Hall Englewood Cliffs, NJ, 1980, vol. 156.

\bibitem{SBD2}
T.~Maehara and K.~Murota, ``Algorithm for error-controlled simultaneous
  block-diagonalization of matrices,'' \emph{SIAM J. Matrix Anal. Appl.},
  vol.~33, no.~2, pp. 605--620, 2011.

\bibitem{Note1}
We computed the UPO by using the continuation tool MATcont \cite
  {dhooge:2003,dhooge2006matcont}.

\bibitem{hunt:1996}
B.~R. Hunt and E.~Ott, ``Optimal periodic orbits of chaotic systems,''
  \emph{Physical review letters}, vol.~76, no.~13, p. 2254, 1996.

\bibitem{sanchez2020exploiting}
R.~J. S{\'a}nchez-Garc{\'\i}a, ``Exploiting symmetry in network analysis,''
  \emph{Communications Physics}, vol.~3, no.~1, pp. 1--15, 2020.

\bibitem{golubitsky2012singularities}
M.~Golubitsky, I.~Stewart, and D.~G. Schaeffer, \emph{Singularities and Groups
  in Bifurcation Theory: Volume II}.\hskip 1em plus 0.5em minus 0.4em\relax
  Springer Science \& Business Media, 2012, vol.~69.

\bibitem{golubitsky2003symmetry}
M.~Golubitsky and I.~Stewart, \emph{The symmetry perspective: from equilibrium
  to chaos in phase space and physical space}.\hskip 1em plus 0.5em minus
  0.4em\relax Springer Science \& Business Media, 2003, vol. 200.

\bibitem{martini:2010}
S.~Martini, M.~Egerstedt, and A.~Bicchi, ``Controllability analysis of
  multi-agent systems using relaxed equitable partitions,'' \emph{International
  Journal of Systems, Control and Communications}, vol.~2, no. 1-3, pp.
  100--121, 2010.

\bibitem{belykh2011}
I.~Belykh and M.~Hasler, ``Mesoscale and clusters of synchrony in networks of
  bursting neurons,'' \emph{Chaos: An Interdisciplinary Journal of Nonlinear
  Science}, vol.~21, no.~1, p. 016106, 2011.

\bibitem{NC}
L.~M. Pecora, F.~Sorrentino, A.~M. Hagerstrom, T.~E. Murphy, and R.~Roy,
  ``Cluster synchronization and isolated desynchronization in complex networks
  with symmetries,'' \emph{Nature Communications}, vol.~5, 2014.

\bibitem{schaub2016graph}
M.~T. Schaub, N.~O'Clery, Y.~N. Billeh, J.-C. Delvenne, R.~Lambiotte, and
  M.~Barahona, ``Graph partitions and cluster synchronization in networks of
  oscillators,'' \emph{Chaos: An Interdisciplinary Journal of Nonlinear
  Science}, vol.~26, no.~9, p. 094821, 2016.

\bibitem{cho2017stable}
Y.~S. Cho, T.~Nishikawa, and A.~E. Motter, ``Stable chimeras and independently
  synchronizable clusters,'' \emph{Physical review letters}, vol. 119, no.~8,
  p. 084101, 2017.

\bibitem{bischer2021simultaneous}
I.~Bischer, C.~D{\"o}ring, and A.~Trautner, ``Simultaneous block
  diagonalization of matrices of finite order,'' \emph{Journal of Physics A:
  Mathematical and Theoretical}, vol.~54, no.~8, p. 085203, 2021.

\bibitem{arapura2004common}
D.~Arapura and C.~Peterson, ``The common invariant subspace problem: an
  approach via gr{\"o}bner bases,'' \emph{Linear algebra and its applications},
  vol. 384, pp. 1--7, 2004.

\bibitem{van1996matrix}
C.~F. Van~Loan and G.~Golub, \emph{Matrix {C}omputations}.\hskip 1em plus 0.5em
  minus 0.4em\relax The Johns Hopkins University Press, 1996, third Ed., Sec.
  12.4.4.

\bibitem{galantai2008subspaces}
A.~Gal{\'a}ntai, ``Subspaces, angles and pairs of orthogonal projections,''
  \emph{Linear and Multilinear Algebra}, vol.~56, no.~3, pp. 227--260, 2008.

\bibitem{kalman1963mathematical}
R.~E. Kalman, ``Mathematical description of linear dynamical systems,''
  \emph{Journal of the Society for Industrial and Applied Mathematics, Series
  A: Control}, vol.~1, no.~2, pp. 152--192, 1963.

\bibitem{blondel2000survey}
V.~D. Blondel and J.~N. Tsitsiklis, ``A survey of computational complexity
  results in systems and control,'' \emph{Automatica}, vol.~36, no.~9, pp.
  1249--1274, 2000.

\bibitem{cline2006computatio}
A.~K. Cline and I.~S. Dhillon, \emph{Computation of the Singular Value
  Decomposition}.\hskip 1em plus 0.5em minus 0.4em\relax CRC Press, jan 2006.

\bibitem{della2020}
F.~Della~Rossa, L.~Pecora, K.~Blaha, A.~Shirin, I.~Klickstein, and
  F.~Sorrentino, ``Symmetries and cluster synchronization in multilayer
  networks,'' \emph{Nature communications}, vol.~11, no.~1, pp. 1--17, 2020.

\bibitem{garcia2015}
M.~I. Garc{\'\i}a~Planas, M.~Planas, and M.~dels Dolors, ``Eigenvectors of
  permutation matrices,'' \emph{Advances in Pure Mathematics}, vol.~5, pp.
  390--394, 2015.

\bibitem{korea}
K.-I. Goh, B.~Kahng, and D.~Kim, ``Universal behavior of load distribution in
  scale-free networks,'' \emph{Phys. Rev. Lett.}, vol.~87, p. 278701, 2001.

\bibitem{nr}
\BIBentryALTinterwordspacing
R.~A. Rossi and N.~K. Ahmed, ``The network data repository with interactive
  graph analytics and visualization,'' in \emph{AAAI}, 2015. [Online].
  Available: \url{https://networkrepository.com}
\BIBentrySTDinterwordspacing

\bibitem{5491276}
R.~D. Zimmerman, C.~E. Murillo-Sánchez, and R.~J. Thomas, ``Matpower:
  Steady-state operations, planning, and analysis tools for power systems
  research and education,'' \emph{IEEE Transactions on Power Systems}, vol.~26,
  no.~1, pp. 12--19, 2011.

\bibitem{capitanescu}
F.~Capitanescu, ``Suppressing ineffective control actions in optimal power flow
  problems,'' \emph{IET Generation, Transmission \& Distribution}, vol.~14,
  no.~13, pp. 2520--2527, 2020.

\bibitem{dhooge:2003}
A.~Dhooge, W.~Govaerts, and Y.~A. Kuznetsov, ``Matcont: a matlab package for
  numerical bifurcation analysis of odes,'' \emph{ACM Transactions on
  Mathematical Software (TOMS)}, vol.~29, no.~2, pp. 141--164, 2003.

\bibitem{dhooge2006matcont}
A.~Dhooge, W.~Govaerts, Y.~A. Kuznetsov, W.~Mestrom, A.~Riet, and B.~Sautois,
  ``Matcont and cl matcont: Continuation toolboxes in matlab,''
  \emph{Universiteit Gent, Belgium and Utrecht University, The Netherlands},
  2006.

\end{thebibliography}

\end{document}